\documentclass[11pt]{article}

\usepackage[letterpaper,margin=1in]{geometry}
\usepackage{amsmath, amssymb, amsthm, amsfonts}
\usepackage{bbm}
\usepackage{bm}

\usepackage{subcaption}

\usepackage{ifthen}
\usepackage{comment}
\usepackage{tikz}
\usetikzlibrary{positioning,decorations.pathreplacing}
\usepackage{comment}
\usepackage{cite}
\usepackage{appendix}
\usepackage{graphicx}
\usepackage{color}
\usepackage{algorithm,xpatch}
\usepackage[noend]{algorithmic}	
\usepackage{epstopdf}
\usepackage{wrapfig}
\usepackage{paralist}
\usepackage[textsize=tiny]{todonotes}

\usepackage{framed}
\usepackage[framemethod=tikz]{mdframed}
\usepackage[bottom]{footmisc}
\usepackage{enumitem}
\setitemize{noitemsep,topsep=3pt,parsep=3pt,partopsep=3pt}
\usepackage[font=small]{caption}
\usepackage{xspace}

\newtheorem{theorem}{Theorem}[section]
\newtheorem{lemma}[theorem]{Lemma}
\newtheorem{meta-theorem}[theorem]{Meta-Theorem}

\newtheorem{corollary}[theorem]{Corollary}

\newtheorem{definition}[theorem]{Definition}

\xpatchcmd{\algorithmic}{\setcounter}{\algorithmicfont\setcounter}{}{}
\providecommand{\algorithmicfont}{}
\providecommand{\setalgorithmicfont}[1]{\renewcommand{\algorithmicfont}{#1}}

\definecolor{darkgreen}{rgb}{0,0.5,0}
\usepackage{hyperref}
\hypersetup{
    unicode=false,          
    colorlinks=true,        
    linkcolor=red,          
    citecolor=darkgreen,        
    filecolor=magenta,      
    urlcolor=cyan           
}
\usepackage[capitalize]{cleveref}

\newcommand{\eps}{\varepsilon}

\DeclareMathOperator{\E}{\mathbb{E}}

\newcommand{\important}[1]{\textbf{#1}}

\newcommand{\exclude}[1]{}

\newcommand{\FullOrShort}{full}

\ifthenelse{\equal{\FullOrShort}{full}}{

  \newcommand{\fullOnly}[1]{#1}
	\newcommand{\tempfullOnly}[1]{#1}
  \newcommand{\shortOnly}[1]{}
  
  }{
	\usepackage{times, microtype}

    \newcommand{\fullOnly}[1]{}
		\newcommand{\tempfullOnly}[1]{}
		\newcommand{\shortOnly}[1]{#1}
    \newcommand{\IncludePictures}[1]{}
   
  }

\newcommand{\ShortOnlyVspace}[1]{\shortOnly{\vspace{#1}}}

\begin{document}
\setalgorithmicfont{\footnotesize}

\date{}

\title{Synchronization Strings: Channel Simulations and Interactive Coding for Insertions and Deletions\footnote{Supported in part by NSF grants CCF-1527110, CCF-1618280 and NSF CAREER award CCF-1750808.\shortOnly{ An extended version of this paper can be found at~\cite{haeupler2017synchronization2:ARXIV}.}}}

\author{
Bernhard Haeupler\\Carnegie Mellon University\\ \texttt{haeupler@cs.cmu.edu} \and
Amirbehshad Shahrasbi\\Carnegie Mellon University\\ \texttt{shahrasbi@cs.cmu.edu}\and
Ellen Vitercik\\Carnegie Mellon University\\ \texttt{vitercik@cs.cmu.edu}}

\maketitle

\begin{abstract}
\global\def\longabstract{
We present many new results related to reliable (interactive) communication over insertion-deletion channels. \emph{Synchronization errors}, such as insertions and deletions, strictly generalize the usual \emph{symbol corruption errors} and are much harder to protect against. 

We show how to hide the complications of synchronization errors in many applications by introducing very general \emph{channel simulations} which efficiently transform an insertion-deletion channel into a regular symbol corruption channel with an error rate larger by a constant factor and a slightly smaller alphabet. We utilize and generalize synchronization string based methods which were recently introduced as a tool to design essentially optimal error correcting codes for insertion-deletion channels. Our channel simulations depend on the fact that, at the cost of increasing the error rate by a constant factor, synchronization strings can be decoded in a streaming manner that preserves linearity of time. Interestingly, we provide a lower bound showing that this constant factor cannot be improved to $1+\eps$, in contrast to what is achievable for error correcting codes. Our channel simulations drastically and cleanly generalize the applicability of synchronization strings.   

We provide new interactive coding schemes which simulate any interactive two-party protocol over an insertion-deletion channel. Our results improve over the interactive coding schemes of Braverman et al.~\cite{braverman2015coding} and Sherstov and Wu~\cite{sherstov2017optimal} which achieve a small constant rate and require exponential time computations with respect to computational and communication complexities. 
We provide the first computationally efficient interactive coding schemes for synchronization errors, the first coding scheme with a rate approaching one for small noise rates, and also the first coding scheme that works over arbitrarily small alphabet sizes. 
%
We also show tight connections between synchronization strings and edit-distance tree codes which allow us to transfer results from tree codes directly to edit-distance tree codes. 

Finally, using on our channel simulations, we provide an explicit low-rate binary insertion-deletion code that improves over the state-of-the-art codes by Guruswami and Wang~\cite{guruswami2017deletion} in terms of \exclude{computational complexity and }rate-distance trade-off.
}

\global\def\shortabstract{
We give many new results related on reliable (interactive) communication over insertion and deletion channels. \emph{Synchronization errors}, such as insertions and deletions, strictly generalize the usual \emph{symbol corruption errors} and are much harder to deal with or code for. 

We show how to hide the complications of synchronization from many applications by giving very general \emph{channel simulations} which efficiently transform an insertion-deletion channel into a regular symbol corruption channel with a slightly larger error rate (and a slightly smaller alphabet). The main tools for this result are synchronization codes, which were recently introduced as a mechanism for designing essentially optimal error correcting codes for insertion-deletion channels. Our channel simulation uses the observation that, at the cost of increasing the error rate by a factor of $6$, synchronization codes can be decoded in a streaming manner which preserves linearity of time. We also provide a lower bound showing that the factor of $6$ can at most be improved to $2$, in contrast to what is achievable for error correcting codes. Our channel simulation drastically and cleanly generalizes the applicability of synchronization codes.   

We provide new interactive coding schemes which simulate any interactive two-party protocol over insertion-deletion channels. Our results improve in many ways over the interactive coding scheme of Braverman et al.~\cite{braverman2015coding} which tolerates an error fraction of up to $1/18-\eps$, achieves a small constant rate, and requires exponential time computations. We provide the first computationally efficient interactive coding schemes for synchronization errors, the first coding scheme with a rate approaching one for small noise rates, and also improve over the maximal tolerable error rate. 
%
We also show tight connections between synchronization codes and edit-distance tree codes which allow us to transfer results from tree codes directly to edit-distance tree codes. 
}

\longabstract

\end{abstract}

\setcounter{page}{0}
\thispagestyle{empty}

\newpage
\section{Introduction}
\ShortOnlyVspace{-2mm}
Communication in the presence of \emph{synchronization errors}, which include both insertions and deletions, is a fundamental problem of practical importance which eluded a strong theoretical foundation for decades. This remained true even while communication in the presence of \emph{half-errors}, which consist of symbol corruptions and erasures, has been the subject of an extensive body of research with many groundbreaking results. Synchronization errors are strictly more general than half-errors, and thus synchronization errors pose additional challenges for robust communication.

In this work, we show that one-way and interactive communication in the presence of synchronization errors can be reduced to the problem of communication in the presence of half-errors. We present a series of efficient channel simulations which allow two parties to communicate over a channel afflicted by synchronization errors as though they were communicating over a half-error channel with only a slightly larger error rate. This allows us to leverage existing coding schemes for robust communication over half-error channels in order to derive strong coding schemes resilient to synchronization errors.

One of the primary tools we use are \emph{synchronization strings}, which were recently introduced by Haeupler and Shahrasbi in order to design essentially optimal error correcting codes (ECCs) robust to synchronization errors \cite{haeupler2017synchronization}. For every $\eps > 0$, synchronization strings allow a sender to index a sequence of messages with an alphabet of size $\eps^{-O(1)}$ in such a way that $k$ synchronization errors are efficiently transformed into $(1+\eps)k$ half-errors for the purpose of designing ECCs. Haeupler and Shahrasbi provide a black-box construction which transforms any ECC into an equally efficient ECC robust to synchronization errors. However, channel simulations and interactive coding in the presence of synchronization errors present a host of additional challenges that cannot be solved by the application of an ECC. Before we describe our results and techniques in detail, we begin with an overview of the well-known interactive communication model.

\paragraph{Interactive communication.} Throughout this work, we study a scenario where there are two communicating parties, whom we call Alice and Bob. The two begin with some input symbols and wish to compute a function of their input by having a conversation. Their goal is to succeed with high probability while communicating as few symbols as possible. In particular, if their conversation would consist of $n$ symbols in the noise-free setting, then they would like to converse for at most $\alpha n$ symbols, for some small $\alpha$, when in the presence of noise. One might hope that Alice and Bob could correspond using error-correcting codes. However, this approach would lead to poor performance because if a party incorrectly decodes a single message, then the remaining communication is rendered useless. Therefore, only a very small amount of noise could be tolerated, namely less than the amount to corrupt a single message.

There are three major aspects of coding schemes for interactive communication that have been extensively studied in the literature. The first is the coding scheme's \important{maximum tolerable error-fraction} or, in other words, the largest fraction of errors for which the coding scheme can successfully simulate any given error-free protocol.
Another important quality of coding schemes for interactive communication, as with one-way communication, is \important{communication rate}, i.e., the amount of communication overhead in terms of the error fraction.
Finally, the \important{efficiency} of interactive coding schemes have been of concern in the previous work.

Schulman initiated the study of error-resilient interactive communication, showing how to convert an arbitrary two-party interactive protocol to one that is robust to a $\delta=1/240$ fraction of adversarial errors with a constant communication overhead \cite{schulman1992communication,schulman1993deterministic}.
Braverman and Rao increased the bound on the tolerable adversarial error rate to $\delta < 1/4$, also with a constant communication overhead \cite{braverman2014toward}. Brakerski et al.~\cite{brakerski2014} designed the first efficient coding scheme resilient to a constant fraction of adversarial errors with constant communication overhead.
The above-mentioned schemes achieve a constant overhead no matter the level of noise. Kol and Raz were the first to study the trade-off between error fraction and communication rate \cite{kol2013interactive}.
Haeupler then provided a coding scheme with a communication rate of $1-O(\sqrt{\delta \log \log (1/\delta)})$ over an adversarial channel \cite{haeupler2014interactive:FOCS}. 
Further prior work has studied coding for interactive communication focusing on communication efficiency and noise resilience~\cite{gelles2015, baverman2014listUnique, haeupler2014optimalII} as well as computational efficiency~\cite{brakerski2012, brakerski2013, brakerski2014, gelles2011, gelles2014, haeupler2014optimalII}. Other works have studied variations of the interactive communication problem~\cite{haeupler2014optimalI, franklin2015, Efremenko2016, agrawal2011, brassard2014}.

The main challenge that \emph{synchronization errors} pose is that they may cause the parties to become ``out of sync.'' For example, suppose the adversary deletes a message from Alice and inserts a message back to her. Neither party will know that Bob is a message behind, and if this corruption remains undetected, the rest of the communication will be useless.
In most state-of-the-art interactive coding schemes for symbol corruptions, the parties communicate normally for a fixed number of rounds and then send back and forth a series of checks to detect any symbol corruptions that may have occurred. One might hope that a synchronization error could be detected during these checks, but the parties may be out of sync while performing the checks, thus rendering them useless as well. Therefore, synchronization errors require us to develop new techniques.

Very little is known regarding coding for interactive communication in the presence of synchronization errors. A 2016 coding scheme by Braverman et al.~\cite{braverman2015coding}, which can be seen as the equivalent of Schulman's seminal result, achieves a small constant communication rate while being robust against a $1/18-\eps$ fraction of errors. The coding scheme relies on edit-distance tree codes, which are a careful adaptation of Schulman's original tree codes \cite{schulman1993deterministic} for edit distance, so the decoding operations are not efficient and require exponential time computations. 
A recent work by Sherstov and Wu~\cite{sherstov2017optimal} closed the gap for maximum tolerable error fraction by introducing a coding scheme that is robust against $1/6-\eps$ fraction of errors which is the highest possible fraction of insertions and deletions under which any coding scheme for interactive communication can work.
Both Braverman et al.~\cite{braverman2015coding} and Sherstov and Wu~\cite{sherstov2017optimal} schemes are of constant communication rate, over large enough constant alphabets, and inefficient. 
In this work we address the next natural questions which, as arose with ordinary corruption interactive communication, are on finding interactive coding schemes that are computationally efficient or achieve super-constant communication efficiency.

\fullOnly{\subsection{Our results}}
\shortOnly{\paragraph{Our results.}}
We present very general channel simulations which allow two parties communicating over an insertion-deletion channel to follow any protocol designed for a regular symbol corruption channel. The fraction of errors on the simulated symbol corruption channel is only slightly larger than that on the insertion-deletion channel. Our channel simulations are made possible by synchronization strings. Crucially, at the cost of increasing the error rate by a constant factor, synchronization strings can be decoded in a streaming manner which preserves linearity of time. 
Note that the similar technique is used in Haeupler and Shahrasbi~\cite{haeupler2017synchronization} to transform  synchronization errors into ordinary symbol corruptions as a stepping-stone to obtain insertion-deletion codes from ordinary error correcting codes in a black-box fashion. However, in the context of error correcting codes, there is no requirement for this transformation to happen in real time. In other words, in the study of insertion-deletion codes by Haeupler and Shahrasbi~\cite{haeupler2017synchronization}, the entire message transmission is done and then the receiving party uses the entire message to transform the synchronization errors into symbol corruptions. In the channel simulation problem, this transformation is required to happen on the fly. 
Interestingly, we have found out that in the harder problem of channel simulation, the factor by which the number of synchronization errors increase by being transformed into corruption errors cannot be improved to $1+o(1)$, in contrast to what is achievable for error correcting codes. This work exhibits the widespread applicability of synchronization strings and opens up several new use cases, such as coding for interactive communication over insertion-deletion channels. Namely, using synchronization strings, we provide techniques to obtain the following simulations of corruption channels over given insertion-deletion channels with binary and large constant alphabet sizes.
\ShortOnlyVspace{-1mm}
\begin{theorem}(Informal Statement of Theorems~\ref{thm:OnewayCnstSizeAlphaSimul}, \ref{thm:InteractiveCnstSizeAlphaSimul}, \ref{thm:NonObliviousGeneralSimulation}, and \ref{thm:BinaryOneWayCompleteSimulation})
\begin{enumerate}
\ShortOnlyVspace{-2mm}
\item[(a)] Suppose that $n$ rounds of a one-way/interactive insertion-deletion channel over an alphabet $\Sigma$ with a $\delta$ fraction of insertions and deletions are given. Using an $\eps$-synchronization string over an alphabet $\Sigma_{syn}$, it is possible to simulate $n\left(1-O_\eps(\delta)\right)$ rounds of a one-way/interactive corruption channel over $\Sigma_{sim}$ with at most $O_\eps\left(n\delta\right)$ symbols corrupted so long as $|\Sigma_{sim}| \times |\Sigma_{syn}| \le |\Sigma|$. 
\item[(b)] Suppose that $n$ rounds of a binary one-way/interactive insertion-deletion channel with a $\delta$ fraction of insertions and deletions are given. It is possible to simulate 
$n(1-\Theta( \sqrt{\delta\log(1/\delta)}))$
 rounds of a binary one-way/interactive corruption channel 
 with $\Theta(\sqrt{\delta\log(1/\delta)})$ fraction of corruption errors between two parties over the given channel.
\end{enumerate}
\end{theorem}

Based on the channel simulations presented above, we present novel interactive coding schemes which simulate any interactive two-party protocol over an insertion-deletion channel. 

We use our large alphabet interactive channel simulation along with constant-rate  efficient coding scheme of Ghaffari and Haeupler~\cite{haeupler2014optimalII} for interactive communication over corruption channels to obtain a coding scheme for insertion-deletion channels that is efficient, has a constant communication rate, and tolerates up to $1/44-\eps$ fraction of errors. Note that despite the fact that this coding scheme fails to protect against the optimal $1/6-\eps$ fraction of synchronization errors as the recent work by Sherstov and Wu~\cite{sherstov2017optimal} does, it is an improvement over all previous work 
in terms of computational efficiency as it is the first efficient coding scheme for interactive communication over insertion-deletion channels. 
\ShortOnlyVspace{-1mm}
\begin{theorem}\label{thm:largeAlphaInteractiveCodingScheme}
For any constant $\eps > 0$ and $n$-round alternating protocol $\Pi$, there is an efficient randomized coding scheme simulating $\Pi$ in presence of $\delta=1/44-\eps$ fraction of edit-corruptions with constant rate (i.e., in $O(n)$ rounds) and in $O(n^5)$ time that works with probability $1-2^{\Theta(n)}$. This scheme requires the alphabet size to be a large enough constant $\Omega_\eps(1)$.
\end{theorem}
\ShortOnlyVspace{-2mm}
Next, we use our small alphabet channel simulation and the corruption channel interactive coding scheme of Haeupler~\cite{haeupler2014interactive:FOCS} to introduce an interactive coding scheme for insertion-deletion channels. This scheme is not only computationally efficient, but also the first with super constant communication rate. In other words, this is the first coding scheme for interactive communication over insertion-deletion channels whose rate approaches one as the error fraction drops to zero.
Our computationally efficient interactive coding scheme achieves a near-optimal communication rate of $1 - O(\sqrt{\delta \log (1/\delta)})$ and tolerates a $\delta$ fraction of errors.
Besides computational efficiency and near-optimal communication rate, this coding scheme improves over all previous work in terms of alphabet size. As opposed to coding schemes provided by the previous work\cite{braverman2015coding, sherstov2017optimal}, our scheme does not require a large enough constant alphabet and works even for binary alphabets.
\ShortOnlyVspace{-1mm}
\begin{theorem}\label{thm:InterFullyAdv}
For sufficiently small $\delta$, there is an efficient interactive coding scheme for fully adversarial binary insertion-deletion channels which is robust against $\delta$ fraction of edit-corruptions, achieves a communication rate of 
$1 - \Theta({\sqrt{\delta\log(1/\delta)}})$, and works with probability $1 - 2^{-\Theta(n\delta)}$.
\end{theorem}


%

We also utilize the channel simulations in one-way settings to provide efficient binary insertion-deletion codes correcting $\delta$-fraction of synchronization errors--for $\delta$ smaller than some constant--with a rate of $1-\Theta(\sqrt{\delta\log(1/\delta)})$. This is an improvement in terms of rate-distance trade-off over the state-of-the-art low-rate binary insertion-deletion codes by Guruswami and Wang~\cite{guruswami2017deletion}. The codes by Guruswami and Wang~\cite{guruswami2017deletion} achieve a rate of $1 - O(\sqrt{\delta} \log(1/\delta))$.

Finally, we introduce a slightly improved definition of edit-distance tree codes\footnote{This improved definition is independently observed by Sherstov and Wu~\cite{sherstov2017optimal}.}, a generalization of Schulman's original tree codes defined by Braverman et al.~\cite{braverman2015coding}. We show that under our revised definition, edit-distance tree codes are closely related to synchronization strings. For example, edit-distance tree codes can be constructed by merging a regular tree code and a synchronization string. This transfers, for example, Braverman's sub-exponential time tree code construction~\cite{braverman2012towards} and the candidate construction of Schulman~\cite{Schulman03:Postscript} from tree codes to edit-distance tree codes.
Lastly, as a side note, we will show that with the improved definition, the coding scheme of Braverman et al.~\cite{braverman2015coding} can tolerate $1/10-\eps$ fraction of synchronization errors rather than $1/18-\eps$ fraction that the scheme based on their original definition did.

\fullOnly{\subsection{The Organization of the Paper}
We start by reviewing basic definitions and concepts regarding interactive communication and synchronization strings in Section~\ref{sec:def}. Then we study channel simulations under various assumptions in Section~\ref{sec:simulation}. We use these channel simulations to obtain novel coding schemes for one-way and interactive communication in Sections~\ref{sec:simulation} and \ref{sec:insDelCodes}. Finally, in Section~\ref{sec:SynchAndTreeCodes}, we discuss connections between synchronization strings, tree codes and edit-distance tree codes introduced by Braverman et al.~\cite{braverman2015coding}.}

\ShortOnlyVspace{-2mm}
\shortOnly{\subsection{Definitions and preliminaries}\label{sec:def}}
\fullOnly{\section{Definitions and preliminaries}\label{sec:def}}
\ShortOnlyVspace{-2mm}
In this section, we define the channel models and communication settings considered in this work. We also provide notation and define synchronization strings.
 
\paragraph{Error model and communication channels.}
In this work, we study two types of channels, which we call \emph{corruption channels} and \emph{insertion-deletion channels.} 
 In the corruption channel model, two parties communicate with an alphabet $\Sigma$, and if one party sends a message $c \in \Sigma$ to the other party, then the other party will receive a message $\tilde{c} \in \Sigma$, but it may not be the case that $c = \tilde{c}$.

In the one-way communication setting over an insertion-deletion channel, messages to the listening party may be inserted, and messages sent by the sending party may be deleted. The two-way channel requires a more careful setup. We emphasize that we cannot hope to protect against arbitrary insertions and deletions in the two-way setting because in the message-driven model, a single deletion could cause the protocol execution to ``hang.''
Therefore, following the standard model of Braverman et al.'s work \cite{braverman2015coding}
that is employed in all other previous works on this problem~\cite{sherstov2017optimal}, we restrict our attention to \emph{edit corruptions}, which consist of a single deletion followed by a single insertion, which may be aimed at either party.
Braverman et al. \cite{braverman2015coding} provide a detailed discussion on their model and show that it is strong enough to generalize other natural models one might consider, including models that utilize clock time-outs to overcome the stalling issue.

In both the one- and two-way communication settings, we study \emph{adversarial channels} with \emph{error rate} $\delta$.
Our coding schemes are robust in both the \emph{fully adversarial} and the \emph{oblivious adversary} models. \shortOnly{We provide the definitions of these standard models in Appendix~\ref{app:prelim}.}
\fullOnly{In the fully adversarial model, the adversary may decide at each round whether or not to interfere based on its state, its own randomness, and the symbols communicated by the parties so far. In the oblivious adversary model, the adversary must decide which rounds to corrupt in advance, and therefore independently of the communication history.
A simple example of an oblivious adversary is the \emph{random error channel}, where each round is corrupted with probability $\delta$.
In models we study, there is no pre-shared randomness between the parties.}

\paragraph{Interactive and one-way communication protocols.} In an \emph{interactive protocol} $\Pi$ over a channel with an alphabet $\Sigma$, Alice and Bob begin with two inputs from $\Sigma^*$ and then engage in $n$ \emph{rounds} of communication. In a single round, each party either listens for a message or sends a message, where this choice and the message, if one is generated, depends on the party's state, its input, and the history of the communication thus far. After the $n$ rounds, the parties produce an output. We study \emph{alternating protocols}, where each party sends a message every other round and listens for a message every other round. In this message-driven paradigm, a party ``sleeps'' until a new message comes, at which point the party performs a computation and sends a message to the other party.
\fullOnly{Protocols in the interactive communication literature typically fall into two categories: \emph{message-driven} and \emph{clock-driven}. In the message-driven paradigm, a party ``sleeps'' until a new message comes, at which point the party performs a computation and sends a message to the other party. Meanwhile, in the clock-driven model, each party has a clock, and during a single tick, each party accepts a new message if there is one, performs a computation, and sends a message to the other party if he chooses to do so. In this work, we study the message-driven model, since in the clock-driven model, dealing with insertions and deletions is too easy. After all, an insertion would mean that one symbol is changed to another as in the case of the standard corruption model and a deletion would be detectable, as it would correspond to an erasure.}
In the presence of noise, we say that a protocol $\Pi'$ \emph{robustly simulates} a deterministic protocol $\Pi$ over a channel $C$ if given any inputs for $\Pi$, the parties can decode the transcript of the execution of $\Pi$ on those inputs over a noise-free channel from the transcript of the execution of $\Pi'$ over $C$. 

Finally, we also study \emph{one-way communication}, where one party sends all messages and the other party listens. Coding schemes in this setting are known as \emph{error-correcting codes}.

\shortOnly{
\paragraph{Synchronization Strings\footnote{See the formal definitions and a review of other technical details in Appendix~\ref{app:prelim}. }.} In short, synchronization strings \cite{haeupler2017synchronization} allow communicating parties to protect against synchronization errors by indexing their messages without blowing up the communication rate.
We describe this technique by introducing two intermediaries, $C_A$ and $C_B$, that conduct the communication over the given insertion-deletion channel. $C_A$ receives all symbols that Alice wishes to send to Bob, $C_A$ sends the symbols to $C_B$, and $C_B$ communicates the symbols to Bob. $C_A$ and $C_B$ handle the synchronization strings and all the extra work that is involved in keeping Alice and Bob in sync by guessing the actual index of symbols received by $C_B$. 
In this way, Alice and Bob communicate via $C_A$ and $C_B$ as though they were communicating over a half-error channel. 

Unfortunately, trivially attaching the indices $1, 2, \dots, n$ to each message will not allow us to maintain a near optimal communication rate. If $C_A$ attaches an index to each of Alice's messages, it would increase the size of $\Sigma$ by a factor of $n$ and the rate would increase by a factor of $1/\log n$, which is far from optimal. Synchronization strings allow the communicating parties to index their messages using an alphabet size that is independent of the total communication length $n$. 

Suppose that with each of Alice's $n$ messages, $C_A$ sends an encoding of her index using a symbol from $\Sigma$. Let $S$ be a ``synchronization string'' consisting of all $n$ encoded indices sent by $C_A$. Further, suppose that the adversary injects a total of $n\delta$ insertions and deletions, thus transforming the string $S$ to the string $S_{\tau}$ (The notation will be justified in Appendix~\ref{app:prelim}). Let some element of $S$ like $S[i]$ pass through the channel without being deleted by the adversary and arrive as $S_{\tau}[j]$. We call $S_{\tau}[j]$ a \emph{successfully transmitted symbol}.

We assume that $C_A$ and $C_B$ know the string $S$ \emph{a priori}. The intermediary $C_B$ will receive a set of transmitted indices $S_{\tau}[1], \dots, S_{\tau}[n]$. Upon receipt of the $j$th transmitted index, $C_B$ guesses the actual index of the received symbol when sent by $C_A$. We call the algorithm that $C_B$ runs to determine this an \emph{($n,\delta)$-indexing algorithm.} The algorithm can also return a symbol $\top$ which represents an ``I don't know'' response. Any successfully transmitted symbols that is decoded incorrectly is called a \emph{misdecoding}. The number of misdecodings that an ($n,\delta)$-indexing algorithm might produced is used as a measure to valuate its quality. An indexing algorithm is \emph{streaming} if its guess for a received symbol only depends on previously arrived symbols.

Haeupler and Shahrasbi defined a family of synchronization strings that admit an $(n,\delta)$-indexing algorithm with strong performance \cite{haeupler2017synchronization}. This family is characterized by a parameter $\eps$ as follows.

\begin{definition}[$\eps$-Synchronization String]\label{def:synCode}
A string $S \in \Sigma^n$ is an $\eps$-synchronization string if for every $1 \leq i < j < k \leq n + 1$ we have that $ED\left(S[i, j),S[j, k)\right) > (1-\eps) (k-i)$. 
\end{definition}

Haeupler and Shahrasbi \cite{haeupler2017synchronization, haeupler2017synchronization3} prove the existence and provide several fast constructions for $\eps$-synchronization strings and provide a streaming $(n,\delta)$-indexing algorithm that returns a solution with $\frac{c_i}{1-\eps} + \frac{c_d\eps}{1-\eps}$ misdecodings. The algorithm runs in time $O(n^5)$, spending $O(n^4)$ on each received symbol.
}

\global\def\Prelim{
\medskip\noindent\textbf{String notation and edit distance.} Let $S$ be a string of $n$ symbols from an alphabet $\Sigma$. For $i,j \in \{1, \dots, n\}$, we denote the substring of $S$ from the $i^{th}$ index through and including the $j^{th}$ index as $S[i,j]$. We refer to the substring from the $i^{th}$ index through, but not including, the $j^{th}$ index as $S[i,j)$. The substrings $S(i,j]$ and $S(i,j)$ are similarly defined. Finally, $S[i]$ denotes the $i^{th}$ symbol of $S$ and $|S| = n$ is the length of $S$.
Occasionally, the alphabets we use are the Cartesian product of several alphabets, i.e. $\Sigma = \Sigma_1 \times \cdots \times \Sigma_n$. If $T$ is a string over $\Sigma,$ then we write $T[i] = \left[a_1, \dots, a_n\right]$, where $a_i \in \Sigma_i$. Throughout this work, we rely on the well-known \emph{edit distance} metric, which measures the smallest number of insertions and deletions necessary to transform one string to another.
\begin{definition}[Edit distance]
The \emph{edit distance} $ED(c,c')$ between two strings $c,c' \in \Sigma^*$ is the minimum number of insertions and deletions required to transform $c$ into $c'$.
\end{definition}
It is easy to see that edit distance is a metric on any set of strings and in particular is symmetric and satisfies the triangle inequality property. Furthermore, $ED\left(c,c'\right) = |c| + |c'| - 2\cdot LCS\left(c,c'\right)$, where $LCS\left(c,c'\right)$ is the longest common substring of $c$ and $c'$.
We also use the \emph{string matching} notation from~\cite{braverman2015coding}:

\begin{definition}[String matching]
Suppose that $c$ \allowbreak and $c'$ are two strings in $\Sigma^*$, and suppose that $*$ is a symbol not in $\Sigma$. Next, suppose that there exist two strings $\tau_1$ and $\tau_2$ in $\left(\Sigma \cup \{*\}\right)^*$ such that $|\tau_1| = |\tau_2|$, $del\left(\tau_1\right) = c$, $del(\tau_2) = c'$, and $\tau_1[i] \approx \tau_2[i]$ for all $i \in \left\{1, \dots, |\tau_1|\right\}$. Here, $del$ is a function that deletes every $*$ in the input string and $a \approx b$ if $a = b$ or one of $a$ or $b$ is $*$. Then we say that $\tau = \left(\tau_1, \tau_2\right)$ is a \emph{string matching} between $c$ and $c'$ (denoted $\tau: c \to c'$). We furthermore denote with $sc\left(\tau_i\right)$ the number of $*$'s in $\tau_i$.
\end{definition}

Note that the \emph{edit distance} between strings $c,c' \in \Sigma^*$ is exactly equal to $\min_{\tau: c \to c'}\left\{sc\left(\tau_1\right) + sc\left(\tau_2\right)\right\}$.

%
%
%

\begin{definition}[Suffix distance]
Given any two strings $c,\tilde{c} \in \Sigma^*$, the \emph{suffix distance} between $c$ and $\tilde{c}$ is $SD\left(c,\tilde{c}\right) = \min_{\tau: c \to \tilde{c}}\left\{\max_{i = 1}^{|\tau_1|} \left\{ \frac{sc\left(\tau_1\left[i,|\tau_1|\right]\right) + sc\left(\tau_2\left[i,|\tau_2|\right]\right)}{|\tau_1| - i + 1 - sc\left(\tau_1\left[i,|\tau_1|\right]\right)}\right\}\right\}$.
\end{definition}

\medskip\noindent\textbf{Synchronization Strings.} We now recall the definition of synchronization strings, which were first introduced by Haeupler and Shahrasbi \cite{haeupler2017synchronization} and further studied in~\cite{haeupler2017synchronization3, haeupler2018synchronization4, cheng2018synchronization}, along with some important lemmas from~\cite{haeupler2017synchronization} which we will need here. In short, synchronization strings \cite{haeupler2017synchronization} allow communicating parties to protect against synchronization errors by indexing their messages without blowing up the communication rate.
We describe this technique by introducing two intermediaries, $C_A$ and $C_B$, that conduct the communication over the given insertion-deletion channel. $C_A$ receives all symbols that Alice wishes to send to Bob, $C_A$ sends the symbols to $C_B$, and $C_B$ communicates the symbols to Bob. $C_A$ and $C_B$ handle the synchronization strings and all the extra work that is involved in keeping Alice and Bob in sync by guessing the actual index of symbols received by $C_B$. 
In this way, Alice and Bob communicate via $C_A$ and $C_B$ as though they were communicating over a half-error channel. 

Unfortunately, trivially adding the indices $1, 2, \dots, n$ to each message will not allow us to maintain a near optimal communication rate. After all, suppose that $C_A$ and $C_B$ are communicating over an alphabet $\Sigma$ of constant size. If $C_A$ adds an index to each of Alice's messages, it would increase the size of $\Sigma$ by a factor of $n$ and the rate would increase by a factor of $1/\log n$, which is far from optimal. Synchronization strings allow the communicating parties to index their messages using an alphabet size that is polynomial in $1/\delta$ and is thus independent of the total communication length $n$. Of course, some accuracy is lost when deviating from the trivial indexing strategy. Therefore, if $C_A$ sends $C_B$ indices attached to each of Alice's $n$ messages, we need a notion of how well $C_B$ is able to determine those indices in the presence of synchronization errors. This is where the notion of a string matching comes in handy.

Suppose that with each of Alice's $n$ messages, $C_A$ sends an encoding of her index using a symbol from $\Sigma$. Let $S$ be a ``synchronization string'' consisting of all $n$ encoded indices sent by $C_A$. Next, suppose that the adversary injects a total of $n\delta$ insertions and deletions, thus transforming the string $S$ to the string $S_{\tau}$. Here, $\tau= (\tau_1, \tau_2)$ is a string matching such that $del(\tau_1) = S$, $del(\tau_2) = S_{\tau}$, and for all $k \in [|\tau_1|] = [|\tau_2|]$, \[(\tau_1[k], \tau_2[k]) = \begin{cases} (S[i], *) &\text{if }S[i] \text{ is deleted}\\
(S[i], S_{\tau}[j]) &\text{if }S_{\tau}[j] \text{ is successfully transmitted and sent as }S[i]\\
(*,S_{\tau}[j]) &\text{if }S_{\tau}[j] \text{ is inserted,}\end{cases}\] where $i$ is the index of $\tau_1[1,k]$ upon deleting the stars in $\tau_1[1,k]$, or in other words, $i = |del(\tau_1[1,k])|$ and similarly $j = |del(\tau_1[1,k])|$. We formally say that a symbol $S_{\tau}[j]$ is \emph{successfully transmitted} if there exists a $k$ such that $|del(\tau_2[1,k])| = j$ and $\tau_1[k] = \tau_2[k]$. It was not inserted or deleted by the adversary.

We assume that $C_A$ and $C_B$ know the string $S$ \emph{a priori}. The intermediary $C_B$ will receive a set of transmitted indices $S_{\tau}[1], \dots, S_{\tau}[n]$. Upon receipt of the $j$th transmitted index, for all $j \in [n]$, $C_B$ approximately matches $S_{\tau}[1,j]$ to a prefix $S[1,i]$ and therefore guesses that $C_A$ has sent $i$ messages. We call the algorithm that $C_B$ runs to determine this matching an \emph{($n,\delta)$-indexing algorithm.} The algorithm can also return a symbol $\top$ which represents an ``I don't know'' response. Formally, we define an ($n,\delta)$-indexing algorithm as follows.

\begin{definition}[$(n,\delta)$-Indexing Algorithm] The pair $(S, \mathcal{D}_S)$ consisting of a string $S \in \Sigma^n$ and an algorithm $\mathcal{D}_S$ is called an $(n,\delta)$-indexing algorithm over alphabet $\Sigma$ if for any set of $n\delta$ insertions and deletions corresponding to the string matching $\tau$ and altering the string $S$ to a string $S_{\tau}$, the algorithm $\mathcal{D}_S(S_{\tau})$ outputs either $\top$ or an index between 1 and $n$ for every symbol in $S_{\tau}$.
\end{definition}

How well can an ($n,\delta)$-indexing algorithm perform? We can only hope to \emph{correctly decode} symbols that are successfully transmitted. Recall that the symbol $S_{\tau}[j]$ is successfully transmitted if there exists an index $k$ such that $|del(\tau_2[1,k])| = j$ and $\tau_1[k] = S[i]$ for some $i \in [n]$. It makes sense, then, to say that the algorithm correctly decodes $S_{\tau}[j]$ if it successfully recovers the index $i$. Indeed, we express this notion formally by saying that an $(n,\delta)$-indexing algorithm $(S, \mathcal{D}_S)$ decodes index $j$ correctly under $\tau = (\tau_1, \tau_2)$ if $\mathcal{D}_S(S_{\tau})$ outputs $i$ and there exists a $k$ such that $i = |del(\tau_1[1,k])|$, $j = |del(\tau_2[1,k])|$, $\tau_1[k] = S[i]$, and $\tau_2[k] = S_{\tau}[j]$. Notice that outputting $\top$ counts as an incorrect decoding.
We now have the language to describe how well an $(n,\delta)$-indexing algorithm performs. An algorithm has at most $k$ misdecodings if for any $\tau$ corresponding to at most $n\delta$ insertions and deletions, there are at most $k$ successfully transmitted, incorrectly decoded indices.
An indexing algorithm is \emph{streaming} if the decoded index
for the $i$th element of the string $S_{\tau}$ only depends on $S_{\tau}[1, i]$.

Haeupler and Shahrasbi defined a family of synchronization strings that admit an $(n,\delta)$-indexing algorithm with strong performance \cite{haeupler2017synchronization}. This family is characterized by a parameter $\eps$ as follows.

\begin{definition}[$\eps$-Synchronization String]\label{def:synCode}
A string $S \in \Sigma^n$ is an $\eps$-synchronization string if for every $1 \leq i < j < k \leq n + 1$ we have that $ED\left(S[i, j),S[j, k)\right) > (1-\eps) (k-i)$. 
\end{definition}

Haeupler and Shahrasbi provide an efficient randomized construction of $\eps$-synchronization strings of arbitrary length. \begin{lemma}[From \cite{haeupler2017synchronization}]\label{lemma:FiniteSyncConstruction}
There exists a randomized algorithm which, for any $\eps >0$, constructs a $\eps$-synchronization string of length $n$ over an alphabet of size $O(\eps^{-4})$ in expected time $O(n^5)$.
\end{lemma}

They prove the following useful property which leads them to an indexing algorithm.

\begin{lemma}[From \cite{haeupler2017synchronization}]
Let $S \in \Sigma^n$ be an $\eps$-synchronization string and let $S_{\tau}[1,j]$ be a prefix of $S_{\tau}$. Then there exists at most one index $i \in [n]$ such that the suffix distance between $S_{\tau}[1,j]$ and $S[1,i]$, denoted by $SD(S_{\tau}[1,j],S[1,i])$ is at most $1-\eps$.
\end{lemma}

This lemma suggests a simple $(n,\delta)$-indexing algorithm given an input prefix $S_{\tau}[1,j]$: Use dynamic programming to search over all prefixes of $S$ for the one with the smallest suffix distance from $S_{\tau}[1,j]$. Haeupler and Shahrasbi present a dynamic programming algorithm which is efficient and results in a small number of misdecodings, and described in the following theorem.

\begin{theorem}[From \cite{haeupler2017synchronization}]\label{thm:RSPDmisdecodings}
Let $S \in \Sigma^n$ be an $\eps$-synchronization string that is sent over an insertion-deletion channel with a $\delta$ fraction of insertions and deletions. There exists a streaming $(n,\delta)$-indexing algorithm that returns a solution with $\frac{c_i}{1-\eps} + \frac{c_d\eps}{1-\eps}$ misdecodings. The algorithm runs in time $O(n^5)$, spending $O(n^4)$ on each received symbol.
\end{theorem}

}\fullOnly{\Prelim}

\section{Channel Simulations}\label{sec:simulation}

In this section, we show how $\eps$-synchronization strings can be used as a powerful tool to simulate corruption channels over insertion-deletion channels.
In Section~\ref{sec:codings}, we use these simulations to introduce coding schemes resilient to insertion-deletion errors.

\fullOnly{We study the context where Alice and Bob communicate over an insertion-deletion channel, but via a blackbox channel simulation, they are able to run coding schemes that are designed for half-error channels. As we describe in Section~\ref{sec:def}, we discuss this simulation by introducing two intermediaries, $C_A$ and $C_B$, that conduct the simulation by communicating over the given insertion-deletion channel.}

\subsection{One-way channel simulation over a large alphabet} \label{sec:OneWayLargeAlphaSimulation}
Assume that Alice and Bob have access to $n$ rounds of communication over a one-way  insertion-deletion channel where the adversary is allowed to insert or delete up to $n\delta$ symbols. In this situation, we formally define a corruption channel simulation over the given insertion-deletion channel as follows:

\begin{definition}[Corruption Channel Simulation]
Let Alice and Bob have access to $n$ rounds of communication over a one-way insertion-deletion channel with the alphabet $\Sigma$. The adversary may insert or delete up to $n\delta$ symbols. 
 Intermediaries $C_A$ and $C_B$ \emph{simulate} $n'$ rounds of a corruption channel with alphabet $\Sigma_{sim}$ over the given channel as follows. First, the adversary can insert a number of symbols into the insertion-deletion channel between $C_A$ and $C_B$. Then for $n'$ rounds $i=1,\dots, n'$, the following procedure repeats:
\begin{enumerate}
\item Alice gives $X_i\in\Sigma_{sim}$ to $C_A$.
\item Upon receiving $X_i$ from Alice, $C_A$ wakes up and sends a number of symbols (possibly zero) from the alphabet $\Sigma$ to $C_B$ through the given insertion-deletion channel. 
The adversary can delete any of these symbols or insert symbols \emph{before}, \emph{among}, or \emph{after} them.
\item Upon receiving symbols from the channel, $C_B$ wakes up and reveals a number of symbols (possibly zero) from the alphabet $\Sigma_{sim}$ to Bob. We say all such symbols are \emph{triggered} by $X_i$.
\end{enumerate}
Throughout this procedure, the adversary can insert or delete up to $n\delta$ symbols. However, $C_B$ is required to reveal exactly $n'$ symbols to Bob regardless of the adversary's actions.
Let $\tilde{X}_1,\cdots, \tilde{X}_{n'}\in\Sigma_{sim}$ be the symbols revealed to Bob by $C_B$.
This procedure successfully simulates $n'$ rounds of a corruption channel with a $\delta'$ fraction of errors if for all but $n'\delta'$ elements $i$ of the set $\{1,\dots,n'\}$, the following conditions hold: 1) $\tilde{X}_i = X_i$; and 2) $\tilde{X}_i$ is triggered by $X_i$.
\end{definition}

When $\tilde{X}_i = X_i$ and $\tilde{X}_i$ is triggered by $X_i$, we call $\tilde{X}_i$ an \emph{uncorrupted symbol}.
The second condition, that $\tilde{X}_i$ is triggered by $X_i$, is crucial to preserving linearity of time, which is the fundamental quality that distinguishes channel simulations
from channel codings.
It forces $C_A$ to communicate each symbol to Alice as soon as it arrives. Studying channel simulations satisfying this condition is especially important in situations where Bob's messages depends on Alice's, and vice versa.

Conditions (1) and (2) also require that $C_B$ conveys at most one uncorrupted symbol each time he wakes up.
As the adversary may delete $n\delta$ symbols from the insertion-deletion channel, $C_B$ will wake up at most $n(1-\delta)$ times.
Therefore, we cannot hope for a corruption channel simulation where Bob receives more than $n(1-\delta)$ uncorrupted symbols. In the following theorem, we prove something slightly stronger: no deterministic one-way channel simulation can guarantee that Bob receives more than $n(1-4\delta/3)$ uncorrupted symbols and if the simulation is randomized, the expected number of uncorrupted transmitted symbols is at most $n(1-7\delta/6)$. This puts channel simulation in contrast to channel coding as one can recover $1-\delta-\eps$ fraction of symbols there (as shown in~\cite{haeupler2017synchronization}).
\shortOnly{The proof is in Appendix~\ref{app:chapter4}.}

%
\begin{theorem}\label{thm:OnewaySimulLowerBound}
Assume that $n$ uses of a one-way insertion-deletion channel over an arbitrarily large alphabet $\Sigma$ with a $\delta$ fraction of insertions and deletions are given. There is no deterministic simulation of a corruption channel over any alphabet $\Sigma_{sim}$ where the simulated channel guarantees more than $n\left(1-4\delta/3\right)$ uncorrupted transmitted symbols. If the simulation is randomized, the expected number of uncorrupted transmitted symbols is at most $n(1-7\delta/6)$.
\end{theorem}

\global\def\OnewaySimulLowerBoundProof{
\shortOnly{\begin{proof}[Proof of Theorem~\ref{thm:OnewaySimulLowerBound}]}
\fullOnly{\begin{proof}}
Consider a simulation of $n'$ rounds of a corruption channel on an insertion-deletion channel. 
Note that for any symbol that $C_A$ receives, she will send some number of symbols to $C_B$. This number can be zero or non-zero and may also depend on the content of the symbol she receives. We start by proving the claim for deterministic simulations. Let $X_1, X_2, \cdots, X_{n'}\in \Sigma_{sim}$ and $X'_1, X'_2, \cdots, X'_{n'} \in \Sigma_{sim}$ be two possible sets of 
inputs that Alice may pass to $C_A$ such that for any $1\le i\le n'$, $X_ i\neq X'_i$ and $Y_1,\cdots,Y_{m}\in\Sigma$ and let $Y'_1,\cdots,Y'_{m'}\in\Sigma$ be the symbols that $C_A$ sends to $C_B$ through the insertion-deletion channel as a result of receiving $\{X_i\}$ and $\{X'_i\}$ respectively.

Now, consider $Y_1, \cdots, Y_{n\delta}$. Let $k$ be the number of $C_A$'s input symbols which are required to trigger her to output $Y_1, \cdots, Y_{n\delta}$. We prove Theorem~\ref{thm:OnewaySimulLowerBound} in the following two cases:
\begin{enumerate}
\item $k \le \frac{2n\delta}{3}$: In this case, $Y_1, \cdots, Y_{n\delta}$ will cause $C_B$ to output at most $\frac{2n\delta}{3}$ uncorrupted symbols. If the adversary deletes $n\delta$ arbitrary elements among $Y_{n\delta+1}, \cdots, Y_{m}$, then $C_B$ will receive $m-2n\delta \le n-2n\delta$ symbols afterwards; Therefore, he cannot output more than $n-2n\delta$ uncorrupted symbols as the result of receiving $Y_{n\delta+1}, \cdots, Y_{m}$. Hence, no simulation can guarantee $n(1-2\delta) + k < n\left(1-\frac{4}{3}\delta\right)$ uncorrupted symbols or more.
\item $k > \frac{2n\delta}{3}$: Consider the following two scenarios: 
\begin{enumerate}
\item[(a)] Alice tries to convey $X_1, X_2, \cdots, X_{n'}$ to Bob using the simulation. The adversary deletes the first $n\delta$ symbols. Therefore, $C_B$ receives $Y_{n\delta+1}, \cdots, Y_{m}$.
\item[(b)] Alice tries to convey $X'_1, X'_2, \cdots, X'_{n'}$ to Bob using the simulation. The adversary inserts $Y_{n\delta+1}, \cdots, Y_{2n\delta}$ at the very beginning of the communication. Therefore, $C_B$ receives $Y_{n\delta+1}, \cdots, Y_{2n\delta}, Y'_1, Y'_2, \cdots, Y'_{m'}$.
\end{enumerate}
Note that the first $n\delta$ symbols that $C_B$ receives in these two scenarios are the same. Assume that $C_B$ outputs $k'$ symbols as the result of the first $n\delta$ symbols he receives. In the first scenario, the number of uncorrupted symbols $C_B$ outputs as the result of his first $n\delta$ inputs is at most $\max\{0, k'-k\}$. Additionally, at most $m-2n\delta \le n-2n\delta$ uncorrupted messages may be conveyed within the rest of the communication. In the second scenario, the number of uncorrupted communicated symbols is at most $n-k'$.

Now, at least for one of these scenarios the number of guaranteed uncorrupted symbols in the simulation is
\begin{eqnarray*}
\min\left\{n-2n\delta+\max\{0,k'-k\}, n-k'\right\} &\le& \max_{k'} \min\left\{n-2n\delta+\max\{0,k'-k\}, n-k'\right\}\\
&\le& \max_{k'} \min\left\{n-2n\delta+\max\{0,k'-2n\delta/3\}, n-k'\right\}\\
 &=& n-4n\delta/3 = n\left(1-\frac{4}{3}\delta\right).
\end{eqnarray*}
\end{enumerate}
This completes the proof for deterministic simulations. Now, we proceed to the case of randomized simulations. 

Take an arbitrary input sequence $X_1,\cdots,X_{n'}\in\Sigma^{n'}$. Let $K_X$ be the random variable that represents the number of $C_A$'s input symbols which are required to trigger her to output her first $n\delta$ symbols to $C_B$.
If $\Pr\{K_{X} \le 2n\delta/3\} \ge \frac{1}{2}$, for any sequence $X_1,\cdots,\bar{X}_{n'}$ given to $C_A$ by Alice, the adversary acts as follows. 
He lets the first $n\delta$ symbols sent by $C_A$ pass through the insertion-deletion channel and then deletes the next $n\delta$ symbols that $C_A$ sends to $C_B$. As in the deterministic case, if $K_X \le 2n\delta/3$, the number of uncorrupted symbols conveyed to Bob cannot exceed $n(1-4\delta/3)$. Hence, the expected number of uncorrupted symbols in the simulation may be upper-bounded by:
\begin{eqnarray*}
\E[\text{Uncorrupted Symbols}] &\le& p\cdot n(1-4\delta/3) + (1-p)\cdot n(1-\delta)\\
&\le& n\left(1-\frac{3+p}{3}\delta\right)\le n\left(1-\frac{3+1/2}{3}\delta\right) = n\left(1-\frac{7}{6}\delta\right).
\end{eqnarray*}

Now, assume that $\Pr\{K_{X} \le 2n\delta/3\} < \frac{1}{2}$. Take an arbitrary input $X'_1, X'_2, \cdots, X'_{n'} \in \Sigma_{sim}$ such that $X_i \neq X'_i$ for all $1\le i \le n'$. Consider the following scenarios:
\begin{enumerate}
\item[(a)] Alice tries to convey $X_1, X_2, \cdots, X_{n'}$ to Bob using the simulation.
The adversary removes the first $n\delta$ symbols sent by $C_A$. This means that $C_B$ receives is $Y_{n\delta+1},\cdots,Y_{m}$ where $Y_{1},\cdots,Y_{m}$ is a realization of $C_A$'s output distribution given $X_1, X_2, \cdots, X_{n'}$ as input.
\item[(b)] Alice tries to convey $X'_1, X'_2, \cdots, X'_{n'}$ to Bob using the simulation. The adversary mimics $C_A$ and generates a sample of $C_A$'s output distribution given $X_1, X_2, \cdots, X_{n'}$ as input. Let that sample be $Y_{1},\cdots,Y_{m}$. The adversary inserts $Y_{n\delta+1},\cdots,Y_{2n\delta}$ at the beginning and then lets the communication go on without errors.
\end{enumerate}

Note that the distribution of the first $n\delta$ symbols that $C_B$ receives, i.e., $Y_{n\delta+1},\cdots,Y_{2n\delta}$, is the same in both scenarios.
Let $K'_{X'}$ be the random variable that represents the number of symbols in $C_B$'s output given  that specific distribution over the symbols $Y_{n\delta+1},\cdots,Y_{2n\delta}$. Now, according to the discussion we had for deterministic simulations, for the first scenario:
\begin{eqnarray*}
\E[\text{Uncorrupted Symbols}] &\le& \E[n-2n\delta+\max\{0,K'_{X'}-K_X\}]\\
&\le& n-2n\delta + p\cdot\E[K'_{X'}] + (1-p)(\E[K'_{X'}]-2n\delta/3)\\
&\le& n-2n\delta + \E[K'_{X'}] - 2(1-p)n\delta/3
\end{eqnarray*}
and for the second one:
$$\E[\text{Uncorrupted Symbols}] \le \E[n-K'_{X'}] \le n-\E[K'_{X'}].$$
Therefore, in one of the above-mentioned scenarios
\begin{eqnarray*}
\E[\text{Uncorrupted Symbols}] &\le& \min\{n-2n\delta + \E[K'_{X'}] - 2(1-p)n\delta/3, n-\E[K'_{X'}]\} \\
&\le& \max_{\gamma}\min\{n-2n\delta + \gamma - 2(1-p)n\delta/3, n-\gamma\}\\
&=& n\left(1-\frac{4-p}{3}\delta\right)\\
&<& n\left(1-\frac{4-1/2}{3}\delta\right)\\
&=& n\left(1-\frac{7}{6}\delta\right).
\end{eqnarray*}

Therefore, for any randomized simulation, there exists an input and a strategy for the adversary where
$$\E[\text{Uncorrupted Symbols}] \le n\left(1-\frac{7}{6}\delta\right).$$
\end{proof}
}
\fullOnly{\OnewaySimulLowerBoundProof}

We now provide a channel simulation using $\eps$-synchronization strings. Every time $C_A$ receives a symbol from Alice (from an alphabet $\Sigma_{sim}$), $C_A$ appends a new symbol from a predetermined $\eps$-synchronization string over an alphabet $\Sigma_{syn}$ to Alice's symbol and sends it as one message through the channel.
On the other side of channel, suppose that $C_B$ has already revealed some number of symbols to Bob. Let $\texttt{I}_B$ be the index of the next symbol $C_B$ expects to receive. \fullOnly{In other words, suppose that $C_B$ has already revealed $\texttt{I}_B - 1$ symbols to Bob.} Upon receiving a new symbol from $C_A$, $C_B$ uses the part of the message coming from the synchronization string to guess the index of the message Alice sent. We will refer to this decoded index as $\tilde{\texttt{I}}_A$ and its actual index as $\texttt{I}_A$.  If $\widetilde{\texttt{I}}_A < \texttt{I}_B$, then $C_B$ reveals nothing to Bob and ignores the message he just received. Meanwhile, if $\widetilde{\texttt{I}}_A = \texttt{I}_B$, then $C_B$ reveals Alice's message to Bob. Finally, if $\widetilde{\texttt{I}}_A > \texttt{I}_B$, then $C_B$ sends a dummy symbol to Bob and then sends Alice's message.


\fullOnly{Given that the adversary can insert or delete up to $n\delta$ symbols, if $C_A$ sends $n$ symbols, then $C_B$ may receive between $n -n\delta$ and $n+ n\delta$ symbols. We do not assume the parties have access to a clock, so we must prevent $C_B$ from stalling after $C_A$ has sent all $n$ messages. Therefore, $C_B$ only listens to the first $n(1-\delta)$ symbols it receives.}

The protocols of $C_A$ and $C_B$ are more formally described in Algorithm~\ref{alg:OnewayC_A} \shortOnly{in Appendix~\ref{app:chapter4}}. Theorem~\ref{thm:OnewayCnstSizeAlphaSimul} details the simulation guarantees\shortOnly{ and the proof is in Appendix~\ref{app:chapter4}}.

\fullOnly{
\begin{algorithm}
\caption{Simulation of a one-way constant alphabet channel}
\begin{algorithmic}[1]\label{alg:OnewayC_A}
\STATE Initialize parameters: $S \leftarrow \eps$-synchronization string of length $n$

\IF {$C_A$}
\STATE Reset Status: $\texttt{i}\leftarrow 0$
\FOR {$n$ iterations}
\STATE Get $m$ from Alice, send $(m, S[\texttt{i}])$ to $C_B$, and increment \texttt{i} by 1.
\ENDFOR
\ENDIF


\IF {$C_B$}
\STATE Reset Status: $\texttt{I}_B\leftarrow 0$
\FOR {$n(1-\delta)$ iterations}
\STATE Receive $(\tilde{m}, \tilde{s})$ sent by $C_A$ and set $\tilde{\texttt{I}}_A \leftarrow \mbox{Synchronization string decode}(\tilde{s}, S)$
\IF {$\tilde{\texttt{I}}_A = \texttt{I}_B$}
\STATE Send $\tilde{m}$ to Bob and increment $\texttt{I}_B$.
\ENDIF
\IF {$\tilde{\texttt{I}}_A < \texttt{I}_B$}
\STATE Continue
\ENDIF
\IF {$\tilde{\texttt{I}}_A > \texttt{I}_B$}
\STATE Send a dummy symbol and then $\tilde{m}$ to Bob, then increment $\texttt{I}_B$ by 2.
\ENDIF
\ENDFOR
\ENDIF
\end{algorithmic}
\end{algorithm}}

%

\begin{theorem}\label{thm:OnewayCnstSizeAlphaSimul}
Assume that $n$ uses of a one-way insertion-deletion channel over an alphabet $\Sigma$ with a $\delta$ fraction of insertions and deletions are given. Using an $\eps$-synchronization string over an alphabet $\Sigma_{syn}$, it is possible to simulate $n(1-\delta)$ rounds of a one-way corruption channel over $\Sigma_{sim}$ with at most $2n\delta(2+(1-\eps)^{-1})$ symbols corrupted so long as $|\Sigma_{sim}| \times |\Sigma_{syn}| \le |\Sigma|$ and $\delta < 1/7$. 
\end{theorem}

\global\def\OnewayCnstSizeAlphaSimulProof{
\shortOnly{\begin{proof}[Proof of Theorem~\ref{thm:OnewayCnstSizeAlphaSimul}]}
\fullOnly{\begin{proof}}
Let Alice and Bob use $n$ rounds of an insertion-deletion channel over alphabet $\Sigma$ as sender and receiver respectively. We describe the simulation as being coordinated by two intermediaries $C_A$ and $C_B$, who act according to Algorithm~\ref{alg:OnewayC_A}.

In order to find a lower-bound on the number of rounds of the simulated communication that remain uncorrupted, we upper-bound the number of rounds that can be corrupted.
To this end, let the adversary insert $k_i$ symbols and delete $k_d$ symbols from the communication. Clearly, the $k_d$ deleted symbols do not pass across the channel. Also, each of the $k_i$ inserted symbols may cause $C_B$ to change $I_B$. We call these two cases \emph{error-bad} incidents. Further, $n\delta+k_i-k_d$ symbols at the end of the communication are not conveyed to Bob as we truncate the communication at $n(1-\delta)$. Moreover, according to Theorem~\ref{thm:RSPDmisdecodings}, $\frac{k_i}{1-\eps} + \frac{k_d\eps}{1-\eps}$ successfully transmitted symbols may be misdecoded upon their arrival. We call such incidents \emph{decoding-bad} incidents. Finally, we need to count the number of successfully transmitted symbols whose indexes are decoded correctly ($\tilde{I}_A=I_A$) but do not get conveyed to Bob because $\tilde{I}_A\not\in \{I_B, I_B+1\}$, which we call \emph{zero-bad} incidents. Zero-bad incidents happen only if $|I_A-I_B|\not=	0$. To count the number of zero-bad incidents, we have to analyze how $I_A$ and $I_B$ change in any of the following cases:

\begin{table*}[h]
\centering
 \begin{tabular}{|| c | c c c||} 
 \hline
 Cases when $I_A > I_B$ & $I_A$ & $I_B$ & $I_A - I_B$ \\ [0.5ex] 
 \hline\hline
Deletion by the Adversary & +1 & 0 & +1 \\
Insertion by the Adversary & 0 & 0,+1,+2 & 0, -1, -2 \\
Correctly Transmitted but Misdecoded & +1 & 0,+1,+2 & -1, 0, +1 \\
Correctly Transmitted and Decoded & +1 & +2 & -1 \\
 \hline
 \hline
 Cases when $I_A < I_B$ & $I_A$ & $I_B$ & $I_B-I_A$ \\ [0.5ex] 
 \hline\hline
Deletion by the Adversary & +1 & 0 & -1 \\
Insertion by the Adversary & 0 & 0,+1,+2 & 0, +1, +2 \\
Correctly Transmitted but Misdecoded & +1 & 0,+1,+2 & -1, 0, +1 \\
Correctly Transmitted and Decoded & +1 & 0 & -1 \\
 \hline
 \hline
 Cases when $I_A = I_B$ & $I_A$ & $I_B$ & $I_B-I_A$ \\ [0.5ex] 
 \hline\hline
Deletion by the Adversary & +1 & 0 & -1 \\
Insertion by the Adversary & 0 & 0,+1,+2 & 0, +1, +2 \\
Correctly Transmitted but Misdecoded & +1 & 0,+1,+2 & -1, 0, +1 \\
Correctly Transmitted and Decoded & +1 & +1 & 0 \\
 \hline
\end{tabular}
\caption{How $I_A$ and $I_B$ change in different scenarios.}\label{tbl:indexingSolutions}
\end{table*}

Note that any insertion may increase $|I_A-I_B|$ by up to 2 units and any misdecoding or deletion may increase $|I_A-I_B|$ by up to 1 unit. 
Therefore, $|I_A-I_B|$ may be increased $2k_i + k_d + \frac{k_i}{1-\eps} + \frac{k_d\eps}{1-\eps}$ throughout the algorithm. 
However, as any successfully transmitted and correctly decoded symbol decreases this variable by at least one, there are at most $2k_i + k_d + \frac{k_i}{1-\eps} + \frac{k_d\eps}{1-\eps}$ zero-bad incidents, i.e. successfully transmitted and correctly decoded symbols that are not conveyed successfully in the simulated corruption channel.

Hence, the following is an upper-bound on the number of symbols that may not remain uncorrupted in the simulated channel:
\begin{eqnarray*}
&&\#(\text{error-bad}) + \#(\text{decoding-bad}) + \#(\text{zero-bad}) + \#(\text{truncated symbols})\\
&\le&k_d + \left[\frac{k_i}{1-\eps} + \frac{k_d\eps}{1-\eps}\right] + \left[2k_i + k_d + 	\frac{k_i}{1-\eps} + \frac{k_d\eps}{1-\eps}\right] + \left[n\delta +k_i - k_d\right]\\
&=& n\delta+ k_d\left(1+\frac{2\eps}{1-\eps}\right)+ k_i\left(3+\frac{2}{1-\eps}\right) \le n\delta\left(4+\frac{2}{1-\eps}\right).
\end{eqnarray*}

As error fraction shall not exceed one, the largest $\delta$ for which this simulation works is as follows.
$$\left.\frac{n\delta\left(4+\frac{2}{1-\eps}\right)}{n(1-\delta)}\right|_{\eps=0} = \frac{6\delta}{1-\delta} < 1 \Leftrightarrow \delta < \frac{1}{7}$$
\end{proof}}

\fullOnly{\OnewayCnstSizeAlphaSimulProof}



\ShortOnlyVspace{-4mm}
\subsection{Interactive channel simulation over a large alphabet}
\ShortOnlyVspace{-1mm}
We now turn to channel simulations for interactive channels. As in Section~\ref{sec:OneWayLargeAlphaSimulation}, we formally define a corruption interactive channel simulation over a given insertion-deletion interactive channel. We then use synchronization strings to present one such simulation.

\begin{definition}[Corruption Interactive Channel Simulation]
Let Alice and Bob have access to $n$ rounds of communication over an interactive insertion-deletion channel with alphabet $\Sigma$. The adversary may insert or delete up to $n\delta$ symbols. 
\fullOnly{The simulation of an interactive corruption channel is performed by a pair of intermediaries $C_A$ and $C_B$ where Alice communicates with $C_A$, $C_A$ interacts over the given insertion-deletion channel with $C_B$, and $C_B$ communicates with Bob. More precisely,} \shortOnly{The intermediaries} $C_A$ and $C_B$ \emph{simulate} $n'$ rounds of a corruption interactive channel with alphabet $\Sigma_{sim}$ over the given channel as follows. 
The communication starts when Alice gives a symbol from $\Sigma_{sim}$ to $C_A$. Then Alice, Bob, $C_A$, and $C_B$ continue the communication as follows:
\begin{enumerate}
\ShortOnlyVspace{-2mm}
\item Whenever $C_A$ receives a symbol from Alice or $C_B$,
he either reveals a symbol from $\Sigma_{sim}$ to  Alice or sends a symbol from $\Sigma$ through the insertion-deletion channel to $C_B$.
\item Whenever $C_B$ receives a symbol from Bob or $C_A$,
he either reveals a symbol from $\Sigma_{sim}$ to Bob or send a symbols from $\Sigma$ through the insertion-deletion channel to $C_A$.
\item Whenever $C_B$ reveals a symbol to Bob, Bob responds with a new symbol from $\Sigma_{sim}$.
\item Whenever $C_A$ reveals a symbol to Alice, Alice responds with a symbol in $\Sigma_{sim}$ except for the $\frac{n'}{2}$th time.
\end{enumerate}
Throughout this procedure, the adversary can inject up to $n\delta$ edit corruptions. However, regardless of the adversary's actions, $C_A$ and $C_B$ have to reveal exactly $n'/2$ symbols to Alice and Bob respectively.

Let $X_1,\dots, X_{n'}$ be the symbols Alice gives to $C_A$ and $\tilde{X}_1,\dots, \tilde{X}_{n'}\in\Sigma_{sim}$ be the symbols $C_B$ reveals to Bob. Similarly, Let $Y_1,\dots, Y_{n'}$ be the symbols Bob gives to $C_B$ and $\tilde{Y}_1,\dots, \tilde{Y}_{n'}\in\Sigma_{sim}$ be the symbols $C_A$ reveals to Alice.
We call each pair of tuples $(X_i, \tilde{X}_i)$ and $(Y_i, \tilde{Y}_i)$ a \emph{round} of the simulated communication. We call a round \emph{corrupted} if its elements are not equal. 
This procedure successfully simulates $n'$ rounds of a corruption interactive channel with a $\delta'$ fraction of errors if for all but $n'\delta'$ of  the rounds are corrupted.
\end{definition}

\fullOnly{\begin{algorithm}[t]
\caption{Simulation of a corruption channel using an insertion-deletion channel with a large alphabet: $C_A$'s procedure}
\begin{algorithmic}[1]\label{alg:Inter_big_A}
\STATE $\Pi \leftarrow n$-round interactive coding scheme over a corruption channel to be simulated
\STATE Initialize parameters: $S \leftarrow \eps$-synchronization string of length $n/2$
\STATE $\texttt{I}_A\leftarrow 0$
\FOR {$n/2 - n\delta\left(1 + \frac{1}{1-\eps}\right)$ iterations}
\STATE Get $m$ from Alice, send $(m, S[\texttt{I}_A])$ to $C_B$, and increment $\texttt{I}_A$ by 1.
\STATE Get $(\tilde{m}, \tilde{S})$ from $C_B$ and send $\tilde{m}$ to Alice.
\ENDFOR
\STATE Commit.
\FOR {$n\delta\left(1 + \frac{1}{1-\eps}\right)$ iterations}
\STATE Send $(0, 0)$ to $C_B$, and increment $\texttt{I}_A$ by 1.
\STATE Get $(\tilde{m}, \tilde{S})$ from $C_B$.
\ENDFOR
\end{algorithmic}
\end{algorithm}}

\fullOnly{\begin{algorithm}[t]
\caption{Simulation of a corruption channel using an insertion-deletion channel with a large alphabet: $C_B$'s procedure}
\begin{algorithmic}[1]\label{alg:Inter_big_B}
\STATE $\Pi \leftarrow n$-round interactive coding scheme over a corruption channel to be simulated
\medskip
\STATE Initialize parameters: $S \leftarrow \eps$-synchronization string of length $n/2$

\medskip

	\STATE $\texttt{I}_B\leftarrow 0$
	\FOR {$n/2$ iterations}
		\STATE Receive $(\tilde{m}, \tilde{s})$ from $C_A$.
		\STATE $\tilde{\texttt{I}}_A \leftarrow \mbox{Synchronization string decode}(\tilde{s}, S)$.
		\IF {Committed} 
		\STATE Send a dummy message to $C_A$.
		\ELSIF {$\tilde{\texttt{I}}_A = \texttt{I}_B$}
			\STATE Send $\tilde{m}$ to Bob and increment $\texttt{I}_B$ by 1.
			\STATE Receive $m$ from Bob and send $(m, 0)$ to $C_A$.
		\ELSIF {$\tilde{\texttt{I}}_A < \texttt{I}_B$}
			\STATE Send a dummy message to $C_A$.
		\ELSIF {$\tilde{\texttt{I}}_A > \texttt{I}_B$}
			\STATE Send a dummy message to Bob.
			\STATE Send $\tilde{m}$ to Bob and increment $\texttt{I}_B$ by 2.\label{step:inc2}
			\STATE Receive $m$ from Bob and send $(m, 0)$ to $C_A$.
		\ENDIF
	\STATE If $\texttt{I}_B =n/2 - n\delta\left(1 + \frac{1}{1-\eps}\right)$, commit.
	\ENDFOR
\end{algorithmic}
\end{algorithm}}

\fullOnly{\begin{figure}
\centering{\includegraphics[width=.28\textwidth]{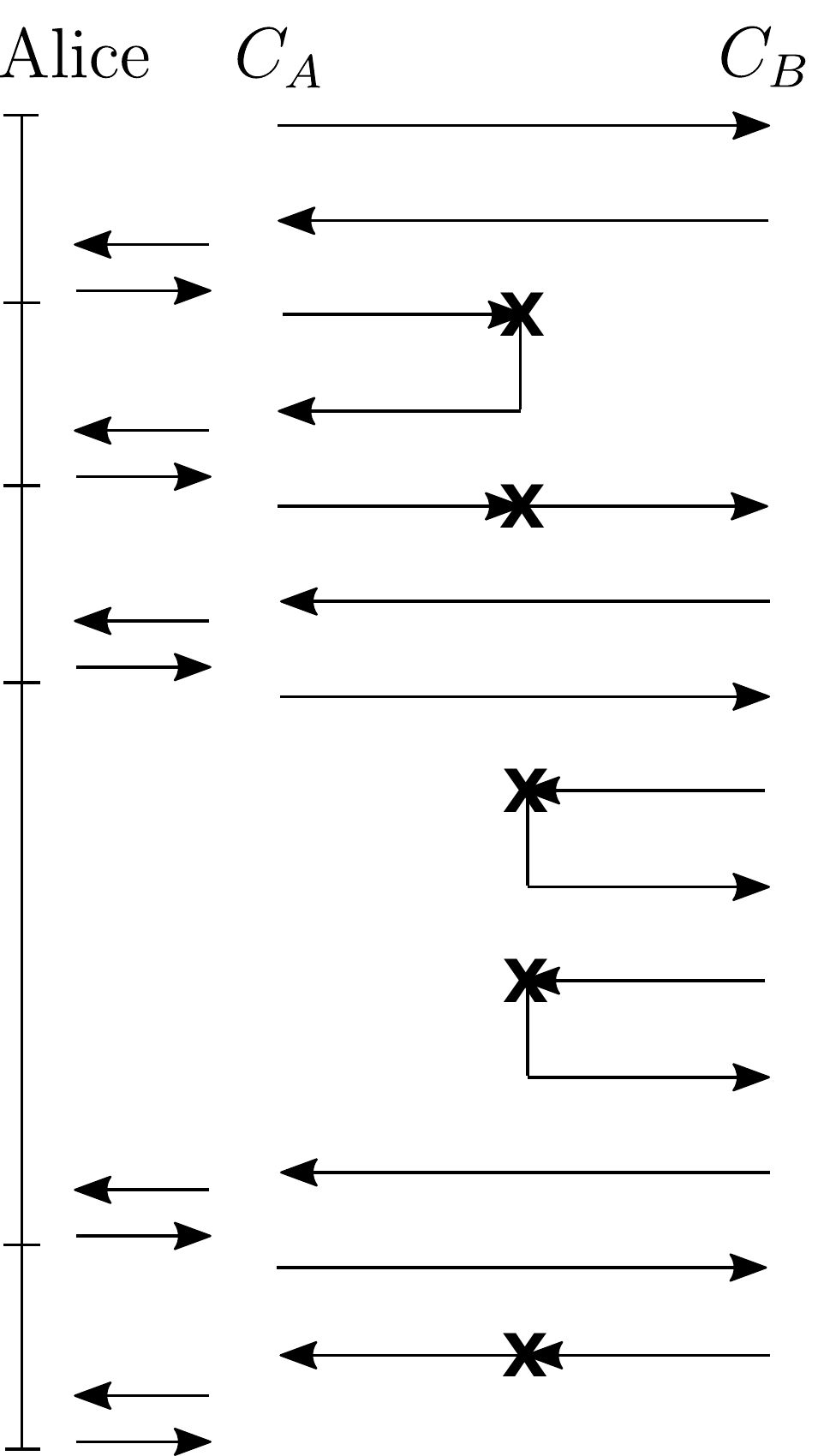}}
\caption{Illustration of five steps where only the first is good.}\label{fig:steps}
\end{figure}}

\shortOnly{The protocol and analysis in this large alphabet setting are similar to the harder case where the alphabet is binary. We cover interactive communication for the binary setting in the next section. Therefore, in the interest of space we defer the pseudo-code and algorithm to Appendix~\ref{app:chapter4}. We prove the following theorem.}

\begin{theorem}\label{thm:InteractiveCnstSizeAlphaSimul}
Assume that $n$ uses of an interactive insertion-deletion channel over an alphabet $\Sigma$ with a $\delta$ fraction of insertions and deletions are given. Using an $\eps$-synchronization string over an alphabet $\Sigma_{syn}$, it is possible to simulate $n-2n\delta(1 + (1-\eps)^{-1})$ uses of an interactive corruption channel over $\Sigma_{sim}$ with at most a $\frac{2\delta(5-3\eps)}{1-\eps +2\eps\delta - 4\delta}$ fraction of symbols corrupted so long as $|\Sigma_{sim}| \times |\Sigma_{syn}| \le |\Sigma|$ and $\delta < 1/14$.
\end{theorem}

\fullOnly{\begin{proof}
Suppose that Alice and Bob want to communicate a total of $n-2n\delta\left(1 + \frac{1}{1-\eps}\right)$ symbols over the simulated corruption channel, and that this channel is simulated by the intermediaries $C_A$ and $C_B$, who are communicating via a total of $n$ uses of an insertion-deletion channel. We will show later on that both parties have the chance to commit before the other has sent $n/2$ messages over the insertion-deletion channel. 
We say an intermediary \emph{commits} when it finishes simulating the channel for its corresponding party, i.e., when it sends the last simulated symbol out.
Intermediaries may commit and yet carry on exchanging symbols over the channel so that the other intermediary finishes its simulation as well. An intermediary may stall by waiting for receiving symbols from the channel but the nature of simulation necessitates the intermediaries not to stall before they commit.

To analyze this simulation, we categorize the bad events that could occur as follows. We say that $C_A$ takes a \emph{step} when it sends a message, receives a message, and completes its required communication with Alice. We say that $C_A$'s step is \emph{good} if the message $C_A$ receives is an uncorrupted response to its previous outgoing message and $C_B$ correctly decodes the index that $C_A$ sends. Figure~\ref{fig:steps} \shortOnly{in Appendix~\ref{app:chapter4}} illustrates a sequence of steps where only the first is good.

If $C_A$ has a good step when $\texttt{I}_A=\texttt{I}_B$ and neither party has committed, then Alice and Bob are guaranteed to have an error-free round of communication. We lower bound the total number of Alice and Bob's error-free rounds of communication by lower bounding the number of good steps that $C_A$ takes when $\texttt{I}_A=\texttt{I}_B$. The total number of good steps $C_A$ takes is $S = n/2 - n\delta\left(1 + \frac{1}{1-\eps}\right) - d$, where $d$ is the number of rounds there are before $C_A$ commits but after $C_B$ commits, if $C_B$ commits first.
If a step is not good, then we say it is \emph{bad}. Specifically, we say that it is \emph{commit-bad} if $C_B$ commits before $C_A$. We say that it is \emph{decoding-bad} if neither party has committed, $C_B$ receives an uncorrupted message from $C_A$, but $C_B$ does not properly decode the synchronization string index. Otherwise, we say that a bad step is \emph{error-bad}. Since every good step corresponds to a message sent by both Alice and Bob, we may lower bound the number of error-free messages sent over the corruption channel by \begin{align*}&2\left(\frac{n}{2} - n\delta\left(1 + \frac{1}{1-\eps}\right) - d - \#(\text{error-bad steps}) - \#(\text{decoding-bad steps})\right. \\
-  & \#(\text{good steps when }\texttt{I}_A \not=\texttt{I}_B)\bigg).\end{align*} In order to lower bound this quantity, we now upper bound $d$, the number of error-bad and decoding-bad steps, and the number of good steps when $\texttt{I}_A \not=\texttt{I}_B$.

We claim that there are at most $n\delta$ error-bad steps since a single adversarial injection could cause $C_A$'s current step to be bad, but that error will not cause the next step to be bad. 
Next, we appeal to Theorem~\ref{thm:RSPDmisdecodings} to bound the number of decoding-bad steps. The cited theorem guarantees that if an $\eps$-synchronization string of length $m$ is sent over an insertion-deletion channel with a $\delta'$ fraction of errors, then the receiver will decode the index of the received symbol correctly for all but $\frac{\delta' m}{1-\eps}$ symbols. In this scenario, $m = n/2$ and $\delta' = n\delta / (n/2) = 2\delta$ since an error over the insertion-deletion channel can cause at most one error in the synchronization string transmission. Therefore, there are at most $n\delta / (1-\eps)$ decoding-bad steps.

Now, a single error, be it an adversarial injection or synchronization strings improper decoding, may cause $|\texttt{I}_B - \texttt{I}_A|$ to increase by at most two since a single error may cause $C_B$ to simulate Step~\ref{step:inc2} of Algorithm~\ref{alg:Inter_big_B} and increment $\texttt{I}_B$ by two, while $\texttt{I}_A$ does not change. Further, if an error does not occur but $C_B$ incorrectly decodes the synchronization symbol $C_A$ sends, then $|\texttt{I}_B - \texttt{I}_A|$ may increase by at most one, since $C_B$ might increase $\texttt{I}_B$ by two or zero while $C_A$ increases $\texttt{I}_A$ by one. Meanwhile, if $|\texttt{I}_A - \texttt{I}_B| \geq 1$ and $C_A$ takes a good step, then this difference will decrease by at least 1. In total, over the course of the computation, since there are at most $n\delta$ adversarial injections and $\frac{n\delta}{1-\eps}$ synchronization string improper decodings, we have that $|\texttt{I}_B - \texttt{I}_A|$ increases by at most $2n\delta + \frac{n\delta}{1-\eps}$, which means that $C_A$ will take at most $2n\delta + \frac{n\delta}{1-\eps}$ good steps where $\texttt{I}_B \not= \texttt{I}_A$. We may bound this number of good steps a bit tighter as follows. Let $\texttt{I}_A^*$ and $\texttt{I}_B^*$ be the values of $\texttt{I}_A$ and $\texttt{I}_B$ when the first of $C_A$ or $C_B$ commits. If $\left|\texttt{I}_A^*-\texttt{I}_B^*\right|>0$, then each party only had at most $2n\delta + \frac{n\delta}{1-\eps} - \left|\texttt{I}_A^*-\texttt{I}_B^*\right|$ good steps where $\texttt{I}_B \not= \texttt{I}_A$.

Finally, we must upper bound the number $d$ of rounds there are before $C_A$ commits but after $C_B$ commits if $C_B$ commits first. Assuming $C_B$ commits first, it must be that $\texttt{I}_B^* = \frac{n}{2} - n\delta\left(1 + \frac{1}{1-\eps}\right)$. Therefore, the number of steps before $C_A$ commits is \[d = \frac{n}{2} - n\delta\left(1 + \frac{1}{1-\eps}\right) -  \texttt{I}_A^* = \left|\texttt{I}_A^*-\texttt{I}_B^*\right|.\] 

We have shown that the number of error-free messages sent over the corruption channel is at least
\begin{align*}
&2\bigg(\frac{n}{2} - n\delta\left(1 + \frac{1}{1-\eps}\right) - d - \#(\text{error-bad steps}) - \#(\text{decoding-bad steps})\\
- &\#(\text{good steps when }\texttt{I}_A \not=\texttt{I}_B)\bigg)\\
\leq \text{ }&2\left(\frac{n}{2} - n\delta\left(1 + \frac{1}{1-\eps}\right) - |\texttt{I}_B^*-\texttt{I}_A^*| - n\delta - \frac{n\delta}{1-\eps} - \left(2n\delta + \frac{n\delta}{1-\eps} - \left|\texttt{I}_A^*-\texttt{I}_B^*\right|\right)\right)\\
= \text{ }&n - 8n\delta - \frac{6n\delta}{1-\eps}. \end{align*} Since Alice and Bob send a total of $n-2n\delta\left(1 + \frac{1}{1-\eps}\right)$ messages over the simulated channel, the error rate $\delta_s$ over the simulated channel is at most $1-\frac{n - 8n\delta - \frac{6n\delta}{1-\eps}}{n-2n\delta\left(1 + \frac{1}{1-\eps}\right)} = \frac{2\delta(5-3\eps)}{1-\eps - 2\delta (1-\eps) - 2\delta}$. Therefore, if $\delta < 1/14$, then $\delta_s < \left.\frac{2\delta(5-3\eps)}{1-\eps +2\delta\eps - 4\delta}\right|_{\eps = 0} = \frac{10\delta}{1 - 4\delta} < 1$, as is necessary.

The last step is to show that $C_A$ has the chance to commit before $C_B$ has sent $n/2$ messages over the insertion-deletion channel, and vice versa. Recall that $C_A$ (respectively $C_B$) commits when $\texttt{I}_A$ (respectively $\texttt{I}_B$) equals $n/2 - n\delta\left(1 + \frac{1}{1-\eps}\right)$. Let $\texttt{i}_B$ be the number of messages sent by $C_B$ over the insertion-deletion channel. The difference $|\texttt{i}_B - \texttt{I}_A|$ only increases due to an error, and a single error can only increase this difference by at most one. Therefore, when $\texttt{i}_B = n/2$, $\texttt{I}_A \geq n/2 - n\delta$, so $C_A$ has already committed. Next, $\texttt{I}_A$ only grows larger than $\texttt{I}_B$ if there is an error or if $C_B$ improperly decodes a synchronization symbol and erroneously chooses to not increase $\texttt{I}_B$. Therefore, $\texttt{I}_A$ is never more than $n\delta\left(1 + \frac{1}{1-\eps}\right)$ larger than $\texttt{I}_B$. This means that when $\texttt{I}_A = n/2$, it must be that $\texttt{I}_B \geq n/2 - \left(1 + \frac{1}{1-\eps}\right)$, so $C_B$ has committed.
\end{proof}}

\subsection{Binary interactive channel simulation}
We now show that with the help of synchronization strings, a binary interactive insertion-deletion channel can be used to simulate a binary interactive corruption channel, inducing a $\widetilde{O}(\sqrt\delta)$ fraction of bit-flips. In this way, the two communicating parties may interact as though they are communicating over a corruption channel. They therefore can employ corruption channel coding schemes while using the simulator as a black box means of converting the insertion-deletion channel to a corruption channel.

The key difference between this simulation and the one-way, large alphabet simulation is that Alice and Bob communicate through $C_A$ and $C_B$ for \emph{blocks} of $r$ rounds, between which $C_A$ and $C_B$ check if they are in sync.
Due to errors, there may be times when Alice and Bob are in disagreement about which block, and what part of the block, they are in. $C_A$ and $C_B$ ensure that Alice and Bob are in sync most of the time.

When Alice sends $C_A$ a message from a new block of communication, $C_A$ holds that message and alerts $C_B$ that a new block is beginning. $C_A$ does this by sending $C_B$ a header that is a string consisting of a single one followed by $s-1$ zeros ($10^{s-1}$). Then, $C_A$ indicates which block Alice is about to start by sending a synchronization symbol to $C_B$. Meanwhile, when $C_B$ receives a $10^{s-1}$ string, he listens for the synchronization symbol, makes his best guess about which block Alice is in, and then communicates with Bob and $C_A$ accordingly. This might entail sending dummy blocks to Bob or $C_A$ if he believes that they are in different blocks. Algorithms~\ref{alg:InterC_A} and \ref{alg:InterC_B} \shortOnly{in Appendix~\ref{app:chapter4}} detail $C_A$ and $C_B$'s protocol. To describe the guarantee that our simulation provides, we first define \emph{block corruption channels}.

\fullOnly{\begin{algorithm}[t]
\caption{Simulation of a corruption channel using an insertion-deletion channel, at $C_A$'s side}
\begin{algorithmic}[1]\label{alg:InterC_A}
\STATE $\Pi \leftarrow n$-round interactive coding scheme over a corruption channel to be simulated \label{step:orig_scheme}
\medskip
\STATE Initialize parameters: $r \leftarrow \sqrt{\frac{\log (1/\delta)}{\delta}}$; $R_{total} \leftarrow \Bigl\lceil n\sqrt{\frac{\delta}{\log (1/\delta)}}\Bigr\rceil$; $s \leftarrow c \log (1/\delta)$; $S \leftarrow \eps$-synchronization string of length $R_{total}$ \label{step:init}
\STATE Reset Status: $\texttt{i}\leftarrow 0$ \label{step:reset_status}

%

\medskip

\FOR {$R_{total}$ iterations}
\STATE Send $s$ zeros to $C_B$
\STATE Send $S[\texttt{i}]$ to $C_B$
\STATE For $r$ rounds, relay messages between $C_B$ and Alice
\STATE $\texttt{i}\leftarrow\texttt{i}+1$
\ENDFOR
\end{algorithmic}
\end{algorithm}

\begin{algorithm}[t]
\caption{Simulation of a corruption channel using an insertion-deletion channel, at $C_B$'s side}
\begin{algorithmic}[1]\label{alg:InterC_B}
\STATE $\Pi \leftarrow n$-round interactive coding scheme over a corruption channel to be simulated \label{step:orig_scheme}
\medskip
\STATE Initialize parameters: $r \leftarrow \sqrt{\frac{\log (1/\delta)}{\delta}}$; $R_{total} \leftarrow \Bigl\lfloor n(1-\delta)\sqrt{\frac{\delta}{\log (1/\delta)}}\Bigr\rfloor$; $s \leftarrow c \log (1/\delta)$; $S \leftarrow \eps$-synchronization string of length $R_{total}$ \label{step:init}
\STATE Reset Status: $\texttt{i}, \texttt{z}, \texttt{I}_B\leftarrow 0$
\medskip

\FOR {$R_{total}$ iterations}
\WHILE {$\texttt{z} < s$}\label{step:wait}
\STATE Receive $b$ from $C_A$
\IF {$b = 0$}
\STATE $\texttt{z}\leftarrow\texttt{z}+1$
\ELSE
\STATE {$\texttt{z} \leftarrow 0$}
\ENDIF
\STATE Send dummy bit to $C_A$
\ENDWHILE
\STATE $\texttt{z} \leftarrow 0$

\medskip

\STATE Receive $m$, the next $|\Sigma_{syn}|$ bits sent by $C_A$ \label{step:sync_symb}
\STATE $\tilde{\texttt{I}}_A \leftarrow \mbox{Synchronization string decode}(m, S)$
\IF {$\tilde{\texttt{I}}_A = \texttt{I}_B$}
\STATE For $r$ rounds, relay messages between $C_A$ and Bob
\STATE $\texttt{I}_B \leftarrow \texttt{I}_B + 1$
\ENDIF
\IF {$\tilde{\texttt{I}}_A < \texttt{I}_B$}\label{step:lt}
\STATE For $r$ rounds, send dummy messages to $C_A$
\ENDIF
\IF {$\tilde{\texttt{I}}_A > \texttt{I}_B$}\label{step:gt}
\STATE For $r$ rounds, send dummy messages to Bob
\STATE For $r$ rounds, relay messages between $C_A$ and Bob
\STATE $\texttt{I}_B \leftarrow \texttt{I}_B + 2$
\ENDIF
\medskip
\ENDFOR
\end{algorithmic}
\end{algorithm}}

\begin{definition}[Block Corruption Channel]
An $n$-round adversarial corruption channel is called a \emph{$(\delta, r)$-block corruption channel} if the adversary is restricted to corrupt $n\delta$ symbols which are covered by $n\delta / r$ blocks of $r$ consecutively transmitted symbols.
\end{definition}

\begin{theorem}\label{thm:InterSim}
Suppose that $n$ rounds of a binary interactive insertion-deletion channel with a $\delta$ fraction of insertions and deletions are given. For sufficiently small $\delta$, it is possible to deterministically simulate 
$n(1-\Theta( \sqrt{\delta\log(1/\delta)}))$
 rounds of a binary interactive
$ (\Theta(\sqrt{\delta\log(1/\delta)}), \sqrt{(1/\delta)\log (1 / \delta)})$-block corruption channel between two parties, Alice and Bob, assuming that 
 all substrings of form $10^{s-1}$ where $s = c\log (1/\delta)$ that Alice sends can be covered by $n\delta$ intervals of $\sqrt{(1/\delta)\log (1 / \delta)}$ consecutive rounds.
The simulation is performed efficiently if the synchronization string is efficient and works.
\end{theorem} 
\shortOnly{
\begin{proof}[Proof Sketch]
Suppose Alice and Bob communicate via intermediaries $C_A$ and $C_B$ who act according to Algorithms~\ref{alg:InterC_A} and \ref{alg:InterC_B}. In total, Alice and Bob will attempt to communicate $n_s$ bits to one another over the simulated channel, while $C_A$ and $C_B$ communicate a total of $n$ bits to one another. The adversary is allowed to insert or delete up to $n\delta$ symbols and $C_A$ sends $n/2$ bits, so $C_B$ may receive between $n/2 -n\delta$ and $n/2+ n\delta$ symbols. To prevent $C_B$ from stalling indefinitely, $C_B$ only listens to the first $n(1-2\delta)/2$ bits he receives.

For $r = \sqrt{(1/\delta)\log (1 / \delta)}$, we define a \emph{chunk} to be $r_c := (s + |\Sigma_{syn}| + r/2)$ consecutive bits that are sent by $C_A$ to $C_B$. In particular, a chunk corresponds to a section header and synchronization symbol followed by $r/2$ rounds of messages sent from Alice. As $C_B$ cares about the first $n(1-2\delta)/2$ bits it receives, there are $\frac{n(1-2\delta)}{2r_c}$ chunks in total.
Hence, 
$n_s = \frac{n(1-2\delta)}{2r_c} \cdot r$
 since $C_B$ and $C_A$'s communication is alternating.

Note that if Alice sends a substring of form $10^{s-1}$ in the information part of a chunk, then Bob mistakenly detects a new block. With this in mind, we say a chunk is \emph{good} if:
\ShortOnlyVspace{-2mm}
\begin{enumerate}
\item No errors are injected in the chunk or affecting $C_B$'s detection of the chunk's header, 
\item $C_B$ correctly decodes the index that $C_A$ sends during the chunk, and 
\item $C_A$ does not send a $10^{s-1}$ substring in the information portion of the chunk. 
\end{enumerate}
\ShortOnlyVspace{-1mm}

If a chunk is not good, we call it \emph{bad}. If the chunk is bad 
because $C_B$ does not decode $C_A$'s index correctly even though they were in sync and no errors were injected, then we call it \emph{decoding-bad}. If it is bad because Alice sends a $10^{s-1}$ substring, we call it \emph{zero-bad} and otherwise, we call it \emph{error-bad}. Throughout the protocol, $C_B$ uses the variable $\texttt{I}_B$ to denote the next index of the synchronization string $C_B$ expects to receive and we use $\texttt{I}_A$ to denote the index of the synchronization string $C_A$ most recently sent. Notice that if a chunk is good and $\texttt{I}_A = \texttt{I}_B$, then all messages are correctly conveyed.

We now bound the maximum number of bad chunks that occur over the course of the simulation. Suppose the adversary injects errors into the $i^{th}$ chunk, making that chunk bad. The $(i+1)^{th}$ chunk may also be bad, since Bob may not be listening for $10^{s-1}$ from $C_A$ when $C_A$ sends them, and therefore may miss the block header. However, if the adversary does not inject any errors into the $(i+1)^{th}$ and the $(i+2)^{th}$ chunk, then the $(i+2)^{th}$ chunk will be good.
In effect, a single error may render at most two chunks useless. Since the adversary may inject $n\delta$ errors into the insertion-deletion channel, this means that the number of chunks that are error-bad is at most $2n \delta$. Additionally, by assumption, the number of zero-bad chunks is also at most $n\delta$.

We also must consider the fraction of rounds that are decoding-bad. In order to do this, we appeal to Theorem~\ref{thm:RSPDmisdecodings}, which guarantees that if an $\eps$-synchronization string of length $N$ is sent over an insertion-deletion channel with a $\delta'$ fraction of insertions and deletions, then the receiver will decode the index of the received symbol correctly for all but $2N\delta'/(1-\eps)$ symbols. 
In this context, $N$ is the number of chunks, i.e. $N = n(1-2\delta)/(2r_c)$, and the fraction of chunks corrupted by errors is $\delta' = 4n\delta/N
$. Therefore, the total number of bad chunks is at most $4 \delta n + 2N\delta'/(1-\eps) = 4 \delta n(3 - \eps)/(1-\eps)$.

In the rest of the proof, which is available in Appendix~\ref{app:chapter4}, we show that all but 
$12\frac{3-\eps}{1-\eps} \cdot \delta n$ 
chunks are good chunks and have $\texttt{I}_A = \texttt{I}_B$ upon their arrival on Bob's side and we conclude that the simulated channel is a
$\left( \frac{3-\eps}{1-\eps} \frac{24\delta}{1 - 2\delta} r_c, r\right)$-block corruption channel. For the asymptotically optimal choice of $r = \sqrt{(1/\delta)\log (1 / \delta)}$, we derive the simulation described in the theorem statement.
\end{proof}
}

\global\def\ProofOfThmInterSim{
\fullOnly{\begin{proof}}
\shortOnly{\begin{proof}[Proof of Theorem~\ref{thm:InterSim}]}
Suppose Alice and Bob communicate via intermediaries $C_A$ and $C_B$ who act according to Algorithms~\ref{alg:InterC_A} and \ref{alg:InterC_B}. In total, Alice and Bob will attempt to communicate $n_s$ bits to one another over the simulated channel, while $C_A$ and $C_B$ communicate a total of $n$ bits to one another. The adversary is allowed to insert or delete up to $n\delta$ symbols and $C_A$ sends $n/2$ bits, so $C_B$ may receive between $n/2 -n\delta$ and $n/2+ n\delta$ symbols. To prevent $C_B$ from stalling indefinitely, $C_B$ only listens to the first $n(1-2\delta)/2$ bits he receives.

For $r = \sqrt{(1/\delta)\log (1 / \delta)}$, we define a \emph{chunk} to be $r_c := (s + |\Sigma_{syn}| + r/2)$ consecutive bits that are sent by $C_A$ to $C_B$. In particular, a chunk corresponds to a section header and synchronization symbol followed by $r/2$ rounds of messages sent from Alice. As $C_B$ cares about the first $n(1-2\delta)/2$ bits it receives, there are $\frac{n(1-2\delta)}{2r_c}$ chunks in total.
Hence, 
$n_s = \frac{n(1-2\delta)}{2r_c} \cdot r$
 since $C_B$ and $C_A$'s communication is alternating.

Note that if Alice sends a substring of form $10^{s-1}$ in the information part of a chunk, then Bob mistakenly detects a new block. With this in mind, we say a chunk is \emph{good} if:
\begin{enumerate}
\item There are no errors injected in the chunk or affecting $C_B$'s detection of the chunk's header, 
\item $C_B$ correctly decodes the index that $C_A$ sends during the chunk, and 
\item $C_A$ does not send a $10^{s-1}$ substring in the information portion of the chunk. 
\end{enumerate}

If a chunk is not good, we call it \emph{bad}. If the chunk is bad 
because $C_B$ does not decode $C_A$'s index correctly even though they were in sync and no errors were injected, then we call it \emph{decoding-bad}. If it is bad because Alice sends a $10^{s-1}$ substring, we call it \emph{zero-bad} and otherwise, we call it \emph{error-bad}. Throughout the protocol, $C_B$ uses the variable $\texttt{I}_B$ to denote the next index of the synchronization string $C_B$ expects to receive and we use $\texttt{I}_A$ to denote the index of the synchronization string $C_A$ most recently sent. Notice that if a chunk is good and $\texttt{I}_A = \texttt{I}_B$, then all messages are correctly conveyed.

We now bound the maximum number of bad chunks that occur over the course of the simulation. Suppose the adversary injects errors into the $i^{th}$ chunk, making that chunk bad. The $(i+1)^{th}$ chunk may also be bad, since Bob may not be listening for $10^{s-1}$ from $C_A$ when $C_A$ sends them, and therefore may miss the block header. However, if the adversary does not inject any errors into the $(i+1)^{th}$ and the $(i+2)^{th}$ chunk, then the $(i+2)^{th}$ chunk will be good.
In effect, a single error may render at most two chunks useless. Since the adversary may inject $n\delta$ errors into the insertion-deletion channel, this means that the number of chunks that are error-bad is at most $2n \delta$. Additionally, by assumption, the number of zero-bad chunks is also at most $n\delta$.

We also must consider the fraction of rounds that are decoding-bad. In order to do this, we appeal to Theorem~\ref{thm:RSPDmisdecodings}, which guarantees that if an $\eps$-synchronization string of length $N$ is sent over an insertion-deletion channel with a $\delta'$ fraction of insertions and deletions, then the receiver will decode the index of the received symbol correctly for all but $2N\delta'/(1-\eps)$ symbols. 
In this context, $N$ is the number of chunks, i.e. $N = n(1-2\delta)/(2r_c)$, and the fraction of chunks corrupted by errors is $\delta' = 4n\delta/N
$. Therefore, the total number of bad chunks is at most $4 \delta n + 2N\delta'/(1-\eps) = 4 \delta n(3 - \eps)/(1-\eps)$.


We will now use these bounds on the number of good and bad chunks to calculate how many errors there are for Alice and Bob, communicating over the simulated channel. As noted, so long as the chunk is good and $\texttt{I}_A = \texttt{I}_B$, then all messages are correctly conveyed to Alice from Bob and vice versa. Meanwhile, a single error, be it an adversarial injection or a synchronization string improper decoding, may cause $|\texttt{I}_B - \texttt{I}_A|$ to increase by at most two since a single error may cause Bob to erroneously simulate Step~\ref{step:sync_symb} and therefore increment $\texttt{I}_B$ by two when, in the worst case, $\texttt{I}_A$ does not change. On the other hand, if $|\texttt{I}_B - \texttt{I}_A| \geq 1$ and the chunk is good, this difference will decrease by at least 1, as is clear from Lines~\ref{step:lt} and \ref{step:gt} of Algorithm~\ref{alg:InterC_B}. 

In total, we have that over the course of the computation, 
$|\texttt{I}_B - \texttt{I}_A|$ increases at most 
$\frac{3 - \eps}{1-\eps} \cdot 4 \delta n$
 times and each time by at most 2.
Therefore, there will be at most 
$\frac{3 - \eps}{1-\eps} \cdot 8 \delta n$ 
good chunks during which $|\texttt{I}_B - \texttt{I}_A| \geq 1$. 	
This gives that all but 
$\frac{3 - \eps}{1-\eps} \cdot 12 \delta n$ 
chunks are good chunks and have $\texttt{I}_A = \texttt{I}_B$ upon their arrival on Bob's side.
Remember that the total number of chunks is $\frac{n(1-2\delta)}{2r_c}$, hence, the simulated channel is a
$\left(\frac{12 \delta n(3 - \eps)/(1-\eps)}{n(1-2\delta)/(2r_c)}, r\right) = \left( \frac{24\delta r_c(3-\eps)}{(1-\eps)(1 - 2\delta)}, r\right)$
block corruption channel.

Thus far, we have shown that one can simulate $n(1-2\delta)\frac{r}{2r_c}$ rounds of a $\left( \frac{24\delta r_c(3-\eps)}{(1-\eps)(1 - 2\delta)}, r\right)$-block corruption channel over a given channel as described in the theorem statement. More specifically, over $n(1 - 2\delta)$ rounds of communication over the insertion-deletion channel, a $\frac{2r_c - r}{2r_c} = \frac{s + |\Sigma_{syn}|}{r_c}$ fraction of rounds are used to add headers and synchronization symbols to chunks and a $\frac{24\delta r_c(3-\eps)}{(1-\eps)(1 - 2\delta)}$ fraction can be lost due to the errors injected by the adversary or $10^{s-1}$ strings in Alice's stream of bits. Therefore, the overall fraction of lost bits in this simulation is
$\frac{s + |\Sigma_{syn}|}{r_c} + \frac{24\delta r_c(3-\eps)}{(1-\eps)(1 - 2\delta)}$. 
Since $s = c\log\frac{1}{\delta}$, $\eps$ and $|\Sigma_{syn}|$ are constants, and $1-2\delta$ approaches to one for small $\delta$, the optimal asymptotic choice is $r = \sqrt
{(1/\delta)\log 1 / \delta}
$. This choice gives a simulated channel with characteristics described in the theorem statement. This simulation is 
performed efficiently if the synchronization symbols are efficiently computable.
\end{proof}}
\fullOnly{\ProofOfThmInterSim}
\ShortOnlyVspace{-2mm}
The simulation stated in Theorem~\ref{thm:InterSim} burdens an additional condition on Alice's stream of bits by requiring it to have a limited number of substrings of form $10^{s-1}$. We now introduce a high probability technique to modify a general interactive communication protocol in a way that makes all substrings of form $10^{s-1}$ in Alice's stream of bits fit into $n\delta$ intervals of length $r=\sqrt{(1/\delta)\log(1/\delta)}$.

\begin{lemma}\label{theorem:InterPreCodeOblivious}
Assume that $n$ rounds of a binary interactive insertion-deletion channel with an oblivious adversary who is allowed to inject $n\delta$ errors are given. There is a pre-coding scheme that can be utilized on top of the simulation introduced in Theorem~\ref{thm:InterSim}. It modifies the stream of bits sent by Alice so that with probability $1 - e^{-\frac{c-3}{2}n\delta\log\frac{1}{\delta}(1+o(1))}$, all substrings of form $10^{s-1}$ where $s = c\log(1/\delta)$ in the stream of bits Alice sends over the simulated channel can be covered by $n\delta$ intervals of length $r=\sqrt{(1/\delta)\log(1/\delta)}$.
This pre-coding scheme comes at the cost of a $\Theta(\sqrt{\delta\log(1/\delta)})$ fraction of the bits Alice sends through the simulated channel.
\end{lemma}

\global\def\InterPreCodeOblivious{
\shortOnly{\begin{proof}[Proof of Lemma~\ref{theorem:InterPreCodeOblivious}]}
\fullOnly{\begin{proof}}
Note that in the simulation process, each $\frac{r}{2}$ consecutive bits Alice sends will form one of the chunks $C_A$ sends to $C_B$ alongside some headers. The idea of this pre-coding is simple. Alice uses the first $\frac{s}{2}$ data bits (and not the header) of each chunk to share $\frac{s}{2}$ randomly generated bits with Bob (instead of running the interactive protocol) and then both of them extract a string $S'$ of $\frac{r}{2}$ $\frac{s}{2}$-wise independent random variables. Then, Alice XORs the rest of data bits she passes to $C_A$ with $S'$ and Bob XORs those bits with $S'$ again to retrieve the original data. 

We now determine the probability that the data block of a chunk of the simulation contains a substring of form $10^{s-1}$. Note that if a block of size $\frac{r}{2}$ contains a $10^{s-1}$ substring, then one of its substrings of length $\frac{s}{2}$ starting at positions $0, \frac{s}{2}, \frac{2s}{2}, \cdots$ is all zero. Since $P$ is $\frac{s}{2}$-wise independent, the probability of each of these $\frac{s}{2}$ substrings containing only zeros is $2^{-\frac{s}{2}} = \delta^\frac{c}{2}$. Taking a union bound over all these substrings, the probability of a block containing a $10^{s-1}$ substring can be bounded above by 
$$p = \frac{r/2}{s/2} \cdot \delta^\frac{c}{2} = \sqrt{\frac{\delta^{c-1}}{\log(1/\delta)}}.$$

Now, we have:
\begin{eqnarray*}
\Pr\left\{\text{Number of blocks containing }10^{s-1} > n\delta \right\} &<& {\frac{n}{s} \choose n\delta} p^{n\delta}\le \left(\frac{ne}{ns\delta}\right)^{n\delta}\left(\sqrt{\frac{\delta^{c-1}}{\log(1/\delta)}}\right)^{n\delta} \\
&<& \left(\frac{\delta^{(c-3)/2}e}{c\log^{3/2}(1/\delta)}\right)^{n\delta} = 
e^{-\frac{c-3}{2}n\delta\log\frac{1}{\delta}(1+o(1))}.
\end{eqnarray*}

\end{proof}
}
\fullOnly{\InterPreCodeOblivious}
\shortOnly{\begin{proof}[Proof sketch]
In the simulation process, each $r/2$ consecutive bits Alice sends forms one of the chunks $C_A$ sends to $C_B$ alongside some headers. The idea of this pre-coding scheme is simple. Alice uses the first $s/2$ data bits (and not the header) of each chunk to share $s/2$ randomly generated bits with Bob (instead of running the interactive protocol) and then both of them extract a string $S'$ of $r/2$ $(s/2)$-wise independent random variables. Then, Alice XORs the rest of data bits she passes to $C_A$ with $S'$ and Bob XORs those bits with $S'$ again to retrieve the original data. In Appendix~\ref{app:chapter4}, we show that this pre-coding scheme guarantees the requirements mentioned in the theorem statement.
\end{proof}}

\ShortOnlyVspace{-2mm}
Applying this pre-coding
 for $c \ge 3$ on top of the simulation from Theorem~\ref{thm:InterSim} implies the following.
\begin{theorem}\label{thm:ObliviousGeneralSimulation}
Suppose that $n$ rounds of a binary interactive insertion-deletion channel with a $\delta$ fraction of insertions and deletions performed by an oblivious adversary are given. 
For sufficiently small $\delta$, it is possible to simulate 
$n(1-\Theta( \sqrt{\delta\log(1/\delta)}))$
 rounds of a binary interactive
$
 (\Theta(\sqrt{\delta\log(1/\delta)}), \sqrt{(1/\delta)\log 1 / \delta})$-block corruption channel between two parties over the given channel.
The simulation works with probability $1 - \exp(-\Theta(n\delta\log (1/\delta)))$ and is efficient if the synchronization string is efficient.
\end{theorem}


\begin{lemma}\label{theorem:InterPreCodeFully}
Suppose that $n$ rounds of a binary, interactive, fully adversarial insertion-deletion channel with a $\delta$ fraction of insertions and deletions are given.
The pre-coding scheme proposed in Lemma~\ref{theorem:InterPreCodeOblivious} ensures that the stream of bits sent by Alice contains fewer than $n\delta$ substrings of form $10^{s-1}$ for $s = c\log(1/\delta)$ and $c>5$ with probability $1 - e^{-\Theta\left(n\delta\log (1/\delta)\right)}$.
\end{lemma}

\global\def\InterPreCodeFullyProof{
\shortOnly{\begin{proof}[Proof of Lemma~\ref{theorem:InterPreCodeFully}]}
\fullOnly{\begin{proof}}
Lemma~\ref{theorem:InterPreCodeOblivious} ensures that for a fixed adversarial error pattern, the stream of bits sent by Alice contains fewer than $n\delta$ substrings of form $10^{s-1}$ upon applying the pre-coding scheme. However, in the fully adversarial setting, the adversary need not fix the error pattern in advance. Since the communication is interactive, the adversary can thus adaptively alter the bits Alice chooses to send. In this proof, we take a union bound over all error patterns with a $\delta$ fraction of errors and show that with high probability, upon applying the pre-coding scheme, the stream of bits sent by Alice contains fewer than $\Theta(n\delta)$ substrings of form $10^{s-1}$.


We claim that the number of error patterns with exactly $k$ insertions or deletions is at most $3^k{n+k\choose k}$. Note that if symbols $s_1, \dots, s_n$ are being sent, each of the $k$ errors can potentially occur within the intervals $[s_1, s_2), [s_2, s_3), \dots, [s_{n-1}, s_n)$, or after $s_n$ is sent. Each error could be a deletion, insertion of ``1'', or insertion of ``0''. This gives the claimed error pattern count. Further, any error pattern with fewer than $k-1$ errors can be thought of as an error pattern with either $k-1$ or $k$ errors where the adversary deletes an arbitrary set of symbols and then inserts the exact same symbols immediately.
Therefore, the number of all possible error patterns with at most $n\delta$ insertions or deletions can be upper-bounded by
$$\sum_{k=n\delta-1}^{n\delta}3^k{n+k\choose k} \le 2\cdot3^{n\delta} {n(1+\delta)\choose n\delta} \le 2\cdot3^{n\delta}\left(\frac{n(1+\delta)e}{n\delta}\right)^{n\delta} < 2\left(\frac{6e}{\delta}\right)^{n\delta} = e^{n\delta\log\frac{1}{\delta}(1+o(1))}.$$

Now, since summation of $\exp\left(n\delta\log(1/\delta)(1+o(1))\right)$ many probabilities any of which smaller than $\exp\left(-\frac{c-3}{2}n\delta\log(1/\delta)(1+o(1))\right)$ is still $\exp\left(-\Theta\left(n\delta\log (1/\delta)\right)\right)$ for $c > 5$, the probability of this pre-coding making more than $n\delta$ disjoint $10^{s-1}$ substrings for $s = c\log\frac{1}{\delta}$ in fully adversarial setting is again $1 - \exp\left(-\Theta\left(n\delta\log (1/\delta)\right)\right)$.
\end{proof}
}
\fullOnly{\InterPreCodeFullyProof}

\global\def\NonObliviousGeneralSimulation{
\begin{theorem}\label{thm:NonObliviousGeneralSimulation}
Suppose that $n$ rounds of a binary interactive insertion-deletion channel with a $\delta$ fraction of insertions and deletions performed by a non-oblivious adversary are given. For a sufficiently small $\delta$, it is possible to simulate 
$n\left(1-\Theta\left( \sqrt{\delta\log(1/\delta)}\right)\right)$
 rounds of a binary interactive
$
 \left(\Theta\left(\sqrt{\delta\log(1/\delta)}\right), \sqrt\frac{\log 1 / \delta}{\delta}\right)$-block corruption channel between two parties, Alice and Bob, over the given channel.
The simulation is efficient if the synchronization string is efficient and works with probability $1 - \exp\left(-\Theta\left(n\delta\log (1/\delta)\right)\right)$.
\end{theorem}
}
\ShortOnlyVspace{-1mm}
Theorem~\ref{thm:InterSim} and Lemma~\ref{theorem:InterPreCodeFully} allow us to conclude that one can perform the simulation stated in Theorem~\ref{thm:InterSim} over any interactive protocol with high probability (see Theorem~\ref{thm:NonObliviousGeneralSimulation}).
\fullOnly{\NonObliviousGeneralSimulation}\shortOnly{Note that one can trivially extend the results of Theorems~\ref{thm:InterSim}~and~\ref{thm:ObliviousGeneralSimulation} to one-way binary communication by ignoring the bits Bob sends. In Appendix~\ref{app:BinaryOneWay}, we discuss the analogies of these theorems in the one-way setting.}

\global\def\BinaryOneWayCommunicationSection{
\subsection{Binary One Way Communication}\label{app:BinaryOneWay}
It is trivial to simplify Algorithms~\ref{alg:InterC_A} and \ref{alg:InterC_B} from Section~\ref{sec:simulation} to prove our simulation guarantees over binary alphabets. Specifically, the messages that Bob sends may be completely ignored and thereby we immediately obtain the following result for one-way insertion-deletion channels:

\begin{theorem}\label{thm:OneWayBinSimul}
Suppose that $n$ rounds of a binary one-way insertion-deletion channel with a $\delta$ fraction of insertions and deletions are given. 
For a sufficiently small $\delta$, it is possible to deterministically simulate 
$n\left(1-\Theta\left( \sqrt{\delta\log(1/\delta)}\right)\right)$
 rounds of a binary 
$$\left(\Theta\left(\sqrt{\delta\log\frac{1}{\delta}}\right), \sqrt \frac{\log (1/\delta)}{\delta}\right)$$
 one-way block corruption channel between two parties, Alice and Bob, over the given channel assuming that all substrings of form $10^{s-1}$ for $s = c\log \frac{1}{\delta}$ in Alice's stream of bits can be covered by $n\delta$ intervals of $\sqrt{\frac{\log(1/\delta)}{\delta}}$ consecutive rounds.
This simulation is performed efficiently if the synchronization string is efficient. 
\end{theorem}

Further, one can use the pre-coding technique introduced in Lemma~\ref{theorem:InterPreCodeOblivious} and Theorem~\ref{thm:OneWayBinSimul} to show that:
\begin{theorem}\label{thm:BinaryOneWayCompleteSimulation}
Suppose that $n$ rounds of a binary one-way insertion-deletion channel with a $\delta$ fraction of insertions and deletions are given. 
For sufficiently small $\delta$, it is possible to simulate 
$n\left(1-\Theta\left( \sqrt{\delta\log(1/\delta)}\right)\right)$
 rounds of a binary 
$\left(\Theta\left(\sqrt{\delta\log(1/\delta)}\right), \sqrt\frac{\log 1 / \delta}{\delta}\right)$
 block corruption channel between two parties, Alice and Bob, over the given channel.
The simulation works with probability $1 - e^{-\Theta\left(n\delta\log(1/\delta)\right)}$, and is efficient if the synchronization string is efficient.
\end{theorem}

\section{Applications: Binary Insertion-Deletion Codes}\label{sec:insDelCodes}

The binary one-way simulation in Theorem~\ref{thm:BinaryOneWayCompleteSimulation} suggests a natural way to overcome insertion-deletion errors in one-way binary channels. One can simply simulate the corruption channel and use appropriate corruption-channel error correcting codes on top of the simulated channel to make the communication resilient to insertion-deletion errors. However, as the simulation works with high probability, this scheme is not deterministic. 

As preserving the streaming quality is not necessary for the sake of designing binary insertion-deletion codes, we can design a deterministic pre-coding and error correcting code that can be used along with the deterministic simulation introduced in Theorem~\ref{thm:OneWayBinSimul} to generate binary insertion-deletion codes.

\begin{lemma}\label{lemma:blockCorCoding}
There exist error correcting codes for 
$\left(\Theta\left(\sqrt{\delta\log(1/\delta)}\right), \sqrt {\log (1/\delta)/\delta}\right)$
block corruption channels with rate $1-\Theta(\sqrt{\delta\log(1/\delta)}) - \delta^{c/2}$ whose codewords are guaranteed to be free of substrings of form $10^{s-1}$ for $c\log\frac{1}{\delta}$.
\end{lemma}

\begin{proof}
Assume that we are sending $n$ bits over a 
$$(p_b, r_b) = \left(\Theta\left(\sqrt{\delta\log\frac{1}{\delta}}\right), \sqrt \frac{\log (1/\delta)}{\delta}\right)$$

block corruption channel. Let us chop the stream of bits into blocks of size $r_b$. Clearly, the fraction of blocks which may contain any errors is at most $2p_b$. Therefore, by looking at each block as a symbol from a large alphabet of size $2^{r_b}$, one can protect them against $2p_b$ fraction of errors by having $\Theta(H_{2^{r_b}}(2p_b) + p_b)$ fraction of redundant blocks. 

In the next step, we propose a way to efficiently transform each block of length $r_b$ of the encoded string into a block of $r_b(1+\delta^{c/2})$ bits so that the resulting stream of bits be guaranteed not to contain $10^{s-1}$ substrings.

To do so, we think of each block of length $r_b$ as a number in the range of zero to $2^{r_b}-1$. Then, we can represent this number in $2^{s/2} - 1$ base and then map each of the symbols of this presentation to strings of $\frac{s}{2}$ bits except the $\frac{s}{2}$ all-zero string. This way, one can efficiently code the string into a stream of bits free of $10^{s-1}$ substrings by losing a $2^{-s/2}$-fraction of bits.

On the other side of the channel, Bob has to split the stream he receives into blocks of length $\frac{s}{2}$, undo the map to find out a possibly corrupted version of the originally encoded message and then decode the $2^{r_b}$-sized alphabet error correcting code to extract Alice's message.

This introduces an insertion-deletion code with rate of $1-\Theta(H_{2^{r_b}}(2p_b) + p_b)-2^{-s/2}$. As $s = c \log\frac{1}{\delta}$, $2^{-s/2} = \delta^{c/2}, p_b = \sqrt{\delta\log\frac{1}{\delta}}$, and

$$H_{2^{r_b}}(2p_b) = \Theta(2p_b \log_{2^{r_b}} (1/2p_b)) = \Theta\left(\frac{2p_b \log_2 (1/2p_b)}{r_b}\right) = \Theta\left(\frac{\sqrt{\delta\log \frac{1}{\delta}}\log\frac{1}{\delta}}{\sqrt{\frac{\log(1/\delta)}{\delta}}}\right) = \Theta\left(\delta\log\frac{1}{\delta}\right)$$
the proof is complete.
\end{proof}

One can set $c \ge 1$ and then apply such error correcting codes on top of a simulated channel as described in Theorem~\ref{thm:OneWayBinSimul} to construct a binary insertion-deletion code resilient to $\delta$ fraction of insertions and deletions with rate $1-\Theta\left(\sqrt{\delta\log\frac{1}{\delta}}\right)$ for a sufficiently small $\delta$.
\begin{theorem}
There exists a constant $0<\delta_0 <1$ such that for any $0<\delta<\delta_0$ there is a binary insertion-deletion code with rate $1-\Theta\left(\sqrt{\delta\log\frac{1}{\delta}}\right)$ which is decodable from $\delta$ fraction of insertions and deletions. 
\end{theorem}
In comparison to the previously known binary insertion-deletion codes\cite{guruswami2016, schulman1999asymptotically}, this code has an improved rate-distance trade-off up to logarithmic factors.
}

\fullOnly{\BinaryOneWayCommunicationSection}
\ShortOnlyVspace{-1mm}
\section{Applications: New Interactive Coding Schemes}\label{sec:codings}
\noindent\textbf{Efficient Coding Scheme Tolerating $1/44$ Fraction of Errors.}\label{sec:largeAlphabetCodingScheme}
In this section, we will provide an efficient coding scheme for interactive communication over insertion-deletion channels by first making use of large alphabet interactive channel simulation provided in Theorem~\ref{thm:InteractiveCnstSizeAlphaSimul} to effectively transform the given channel into a simple corruption interactive channel and then use the efficient constant-rate coding scheme of Ghaffari and Haeupler~\cite{haeupler2014optimalII} on top of the simulated channel. This will give an efficient constant-rate interactive communication over large enough constant alphabets as described in Theorem~\ref{thm:largeAlphaInteractiveCodingScheme}. We review the following theorem of Ghaffari and Haeupler~\cite{haeupler2014optimalII}
before proving Theorem~\ref{thm:largeAlphaInteractiveCodingScheme}.

\begin{theorem}[Theorem 1.1 from \cite{haeupler2014optimalII}]\label{GhaffariHaeupler2014Scheme}
For any constant $\eps > 0$ and $n$-round protocol $\Pi$ there is a randomized non-adaptive coding scheme that robustly simulates $\Pi$ against an adversarial error rate of $\rho \le 1/4 - \eps$ using $N = O(n)$ rounds, a near-linear $n\log^{O(1)} n$ computational complexity, and failure probability $2^{-\Theta(n)}$.
\end{theorem}

\begin{proof}[Proof of Theorem~\ref{thm:largeAlphaInteractiveCodingScheme}]
For a given insertion-deletion interactive channel over alphabet $\Sigma$ suffering from $\delta$ fraction of edit-corruption errors, Theorem~\ref{thm:InteractiveCnstSizeAlphaSimul} enables us to simulate 
$n-2n\delta(1+(1-\eps')^{-1})$
 rounds of ordinary interactive channel with 
 $\frac{2\delta(5-3\eps')}{1-\eps' +2\eps'\delta - 4\delta}$
  fraction of symbol by designating $\log |\Sigma_{syn}|$ bits of each symbol to index simulated channel's symbols with an $\eps'$-synchronization string over $\Sigma_{syn}$.

One can employ the scheme of Ghaffari and Haeupler~\cite{haeupler2014optimalII} over the simulated channel as long as error fraction is smaller than $1/4$. 
Note that 
$\left.\frac{2\delta(5-3\eps')}{1-\eps' +2\delta\eps' - 4\delta}\right|_{\eps' = 0} = \frac{10\delta}{1 - 4\delta} < \frac{1}{4} \Leftrightarrow \delta < \frac{1}{44}.$
Hence, as long as $\delta = 1/44-\eps$ for $\eps > 0$, for small enough $\eps' = O_\eps(1)$, the simulated channel has an error fraction that is smaller than $1/4$. Therefore, by running the efficient coding scheme of Theorem~\ref{GhaffariHaeupler2014Scheme} over this simulated channel one gets a constant rate coding scheme for interactive communication that is robust against $1/44-\eps$ fraction of edit-corruptions. Note that this simulation requires the alphabet size to be large enough to contain synchronization symbols (which can come from a polynomially large alphabet in terms of $\eps'$ according to Lemma~\ref{lemma:FiniteSyncConstruction}) and also meet the alphabet size requirements of Theorem~\ref{GhaffariHaeupler2014Scheme}. This requires the alphabet size to be $\Omega_\eps(1)$, i.e., a large enough constant merely depending on $\eps$. The success probability and time complexity are direct consequences of Theorem~\ref{GhaffariHaeupler2014Scheme} and Theorem~\ref{thm:RSPDmisdecodings}.
\end{proof}

\medskip
\noindent\textbf{Efficient Coding Scheme with Near-Optimal Rate over Small Alphabets.}
In this section we present another insertion-deletion interactive coding scheme that achieves near-optimal communication efficiency as well as computation efficiency by employing a similar idea as in Section~\ref{sec:largeAlphabetCodingScheme}.

\shortOnly{In order to derive a rate-efficient interactive communication coding scheme over small alphabet insertion-deletion channels, Algorithms~\ref{alg:InterC_A} and \ref{alg:InterC_B} can be used to simulate a corruption channel and then the rate-efficient interactive coding scheme for corruption channels introduced by Haeupler \cite{haeupler2014interactive:FOCS} can be used on top of the simulated channel. 
In Appendix~\ref{app:chapter5-interactive}, we provide a quick review of Haeupler's approach alongside its performance over block corruption channels and the proof of the following results.}



\global\def\HaeuplerFourteenRecap{
In order to derive a rate-efficient interactive communication coding scheme over small alphabet insertion-deletion channels, Algorithms~\ref{alg:InterC_A} and \ref{alg:InterC_B} can be used to simulate a corruption channel and then the rate-efficient interactive coding scheme for corruption channels introduced by Haeupler~\cite{haeupler2014interactive:FOCS} can be used on top of the simulated channel. 

We start by a quick review of Theorems 7.1 and 7.2 of Haeupler~\cite{haeupler2014interactive:FOCS}.

\begin{theorem}[Theorem 7.1 from \cite{haeupler2014interactive:FOCS}]\label{thm:H14Thm7.1}
Suppose any $n$-round protocol $\Pi$ using any alphabet $\Sigma$. Algorithm 3 [from \cite{haeupler2014interactive:FOCS}] is a computationally efficient randomized coding scheme which given $\Pi$, with probability $1-2^{-\Theta(n\delta)}$, robustly simulates it over any oblivious adversarial error channel with alphabet $\Sigma$ and error rate $\delta$. The simulation uses $n(1 + \Theta(\sqrt{\delta}))$ rounds and therefore achieves a communication rate of $1-\Theta(\sqrt{\delta})$.
\end{theorem}

\begin{theorem}[Theorem 7.2 from \cite{haeupler2014interactive:FOCS}]
\label{thm:H14Thm7.2}
Suppose any $n$-round protocol $\Pi$ using any alphabet $\Sigma$. Algorithm 4 [from \cite{haeupler2014interactive:FOCS}] is a computationally efficient randomized coding scheme which given $\Pi$, with probability $1-2^{-\Theta(n\delta)}$, robustly simulates it over any fully adversarial error channel with alphabet $\Sigma$ and error rate $\delta$. The simulation uses $n(1 + \Theta\left(\sqrt{\delta\log\log \frac{1}{\delta}}\right)$ rounds and therefore achieves a communication rate of $1-\Theta\left(\sqrt{\delta\log\log\frac{1}{\delta}}\right)$.
\end{theorem}

The interaction between the error rate and rate loss provided in Theorems~\ref{thm:H14Thm7.1} and~\ref{thm:H14Thm7.2} (Theorems 7.1 and 7.2 of \cite{haeupler2014interactive:FOCS}) leads us to the following corollary.
\begin{corollary}\label{obs:interactive}
There are high-probability efficient coding schemes for interactive communication over insertion-deletion channels that are robust against $\delta$ fraction of edit-corruptions for sufficiently small $\delta$ and have communication rate of $1 - \Theta\left({\sqrt[4]{\delta\log\frac{1}{\delta}}}\right)$ against oblivious adversaries and 
$1 - \Theta\left({\sqrt[4]{\delta\log\frac{1}{\delta} \log^2\log\frac{1}{\delta} }}\right)$
in fully adversarial setup.
\end{corollary}

However, these results can be improved upon by taking a closer look at the specifics of the interactive communication coding scheme in \cite{haeupler2014interactive:FOCS}.

In a nutshell, the interactive coding scheme proposed in \cite{haeupler2014interactive:FOCS} simulates an interactive protocol $\Pi$ by splitting the communication into iterations. 
In each iteration, the coding scheme lets Alice and Bob communicate for a block of $\bar{r}$ rounds, then uses $r_c$ rounds of communication after each block so the parties can verify if they are on the same page and then decide whether to continue the communication or to roll back the communication for some number of iterations. The parties perform this verification by exchanging a fingerprint (hash) of their versions of the transcript. Next, each party checks if the fingerprint he receives matches his own, which in turn identifies whether the parties agree or disagree about the communication transcript. Based on the result of this check, each party decides if he should continue the communication from the point he is already at or if he should roll back the communication for some number of iterations and continue from there. (see Algorithms 3 and 4 of \cite{haeupler2014interactive:FOCS})

The analysis in Section 7 of Heaupler \cite{haeupler2014interactive:FOCS} introduces a potential function $\Phi$ which increases by at least one whenever a round of communication is free of any errors or hash collisions. A \emph{hash collision} occurs if the parties' transcripts do not agree due to errors that occurred previously, yet the two parties' fingerprints erroneously match. The analysis also shows that whenever a round of communication contains errors or hash collisions, regardless of the number of errors or hash collisions happening in a round, the potential drops by at most a fixed constant. (Lemmas 7.3 and 7.4 of \cite{haeupler2014interactive:FOCS})

For an error rate of $\bar\delta$, there can be at most $n\bar\delta$ rounds suffering from an error. Haeupler \cite{haeupler2014interactive:FOCS} shows that the number of hash collisions can also be bounded by $\Theta(n\bar\delta)$ with exponentially high probability. (Lemmas 7.6 and 7.7 and Corollary 7.8 of \cite{haeupler2014interactive:FOCS})
Given that the number of errors and hash collisions is bounded by $\Theta(n\bar\delta)$, Haeupler \cite{haeupler2014interactive:FOCS} shows that if $\Phi > \frac{n}{\bar{r}} + \Theta(n\bar\delta)$, then the two parties will agree on the first $n$ steps of communication and therefore the communication will be simulated thoroughly and correctly. Therefore, after $\frac{n}{\bar{r}} + \Theta(n\bar\delta)$ rounds the simulation is complete, and the rate of this simulation is $1 - \Theta\left(r\bar\delta + \frac{r_c}{\bar{r}}\right)$.
}

\fullOnly{\HaeuplerFourteenRecap}

\begin{theorem}[Interactive Coding against Block Corruption]\label{thm:InteractiveCodeBlock}
By choosing an appropriate block length in the Haeupler \cite{haeupler2014interactive:FOCS} coding scheme for oblivious adversaries (Theorem~\ref{thm:H14Thm7.1}), one obtains a robust efficient interactive coding scheme for $(\delta_b, r_b)$-block corruption channel with communication rate
$1-\Theta(\sqrt{\delta_b\max\left\{\delta_b, 1/r_b\right\}})$
that works with probability $1-2^{-\Theta(n\delta_b/r_b )}$.
%
%
\end{theorem}
\global\def\ProofOfThmInteractiveCodeBlock{
\shortOnly{\begin{proof}[Proof of Theorem~\ref{thm:InteractiveCodeBlock}]}
\fullOnly{\begin{proof}}
Let us run Haeupler \cite{haeupler2014interactive:FOCS} scheme with block size $\bar{r}$. This way, each block of corruption that rises in channel, may corrupt 
$\max\left\{\frac{r_b}{\bar{r}}, 1\right\}$
blocks of Haeupler \cite{haeupler2014interactive:FOCS} scheme. Therefore, the total number of corrupted blocks in Haeupler \cite{haeupler2014interactive:FOCS} scheme can be:
$$\frac{n\delta_b}{r_b}\max\left\{\frac{r_b}{\bar{r}}, 1\right\}$$
Therefore, the total fraction of Haeupler \cite{haeupler2014interactive:FOCS} scheme's blocks containing errors is at most 
$$\bar{\delta} = \delta_b\max\left\{1, \frac{\bar{r}}{r_b}\right\}$$

For oblivious adversaries, Haeupler \cite{haeupler2014interactive:FOCS} suggests a verification process which can be performed in $r_c = \Theta(1)$ steps at the end of each round. We use the exact same procedure. Lemma 7.6 of \cite{haeupler2014interactive:FOCS} directly guarantees that the fraction of hash collisions using this procedure is upper-bounded by $\Theta(\bar\delta)$ with probability $1 - 2^{-\Theta(n\bar{\delta}/\bar{r})}$. As the fraction of blocks suffering from hash collisions or errors is at most $\Theta(\bar{\delta})$, the communication can be shown to be complete in $(1+\bar{\delta})$ multiplicative factor of rounds by the same potential argument as in \cite{haeupler2014interactive:FOCS}.

Therefore, the rate lost in this interactive coding scheme is \[\Theta\left(\bar\delta + \frac{r_c}{\bar r}\right) = \Theta\left(\delta_b\max\left\{1, \frac{\bar{r}}{r_b}\right\} + \frac{1}{\bar r}\right)\]

Now, the only remaining task is to find $\bar r$ that minimizes the rate loss mentioned above.
If we choose $\bar r \le r_b$, the best choice is $\bar r = r_b$ as it reduces the rate to 
$\delta_b + \frac{1}{r_b}$.
On the other hand if we choose $\bar r \ge  r_b$, the optimal choice is $\bar r = \sqrt{\frac{r_b}{\delta_b}}$ if $\sqrt{\frac{r_b}{\delta_b}} \ge r_b \Leftrightarrow r_b \le \frac{1}{\delta_b}$ or $\bar r = r_b$ otherwise. 
Hence, we set
$$
\bar r = \Bigg\{
\begin{tabular}{cc}
$\sqrt{\frac{r_b}{\delta_b}}$ & if $r_b \le \frac{1}{\delta_b}$ \\[2mm]
$r_b$ & $r_b > \frac{1}{\delta_b}$
\end{tabular}\\
$$
Plugging this values for $\bar r$ gives that:
$$
Rate = \Bigg\{
\begin{tabular}{cc}
$1-\Theta\left(\sqrt{\frac{\delta_b}{r_b}}\right)$ & $r_b \le \frac{1}{\delta_b}$ \\[2mm]
$1-\Theta\left(\delta_b\right)$ & $r_b > \frac{1}{\delta_b}$
\end{tabular} = 1-\Theta\left(\sqrt{\delta_b\max\left\{\delta_b, \frac{1}{r_b}\right\}}\right)\\
$$
Also, the probability of this coding working correctly is:
$$
1 - 2^{-\Theta(n\bar{\delta}/\bar{r})} = 1-2^{-\delta_b\max\left\{\frac{1}{\bar{r}} , \frac{1}{r_b}\right\}}= \Bigg\{
\begin{tabular}{cc}
$1-2^{-\delta_b\max\left\{\sqrt\frac{\delta_b}{r_b}, \frac{1}{r_b}\right\}}$ & $r_b \le \frac{1}{\delta_b}$ \\[2mm]
$1-2^{-\delta_b\max\left\{\frac{1}{r_b}, \frac{1}{r_b}\right\}}$ & $r_b > \frac{1}{\delta_b}$
\end{tabular} = 1-2^{-\Theta(\frac{\delta_b}{r_b}n)}\\
$$
\end{proof}
}
\fullOnly{\ProofOfThmInteractiveCodeBlock}

Applying the coding scheme of Theorem~\ref{thm:InteractiveCodeBlock} over the simulation from Theorem~\ref{thm:ObliviousGeneralSimulation} implies the following.

\begin{theorem}\label{thm:InterOblvAdv}
For sufficiently small $\delta$, there is an efficient interactive coding scheme over binary insertion-deletion channels which, is robust against $\delta$ fraction of edit-corruptions by an oblivious adversary, achieves a communication rate of $1 - \Theta({\sqrt{\delta\log(1/\delta)}})$, and works with probability $1 - 2^{-\Theta(n\delta)}$.
\end{theorem}

Moreover, \shortOnly{in Appendix~\ref{app:chapter5-interactive}, }we show that this result is extendable for the fully adversarial setup, as summarized in Theorem~\ref{thm:InterFullyAdv}.

\global\def\ProofOfInterFullyAdv{
\begin{proof}[Proof of Theorem~\ref{thm:InterFullyAdv}]
Similar to the proofs of Theorems~\ref{thm:InteractiveCodeBlock}~and~\ref{thm:InterOblvAdv}, we use the simulation discussed in Theorem~\ref{thm:NonObliviousGeneralSimulation} and coding structure introduced in Haeupler \cite{haeupler2014interactive:FOCS} with rounds of length $\bar{r} = \sqrt \frac{\log ({1}/{\delta})}{\delta}$ on top of that. However, this time, we use another verification strategy with length of $r_c = \log\log\frac{1}{\delta}$ which is used in Haeupler \cite{haeupler2014interactive:FOCS} as a verification procedure in the interactive coding scheme for fully adversarial channels.

The idea of this proof, similar to the proof of Theorem \ref{thm:InteractiveCodeBlock}, is to show that the number of rounds suffering from errors or hash collisions can be bounded by $\Theta(n\delta)$ with high probability and then apply the same potential function argument. All of the steps of this proof are, like their analogs in Theorem \ref{thm:InteractiveCodeBlock}, implications of Haeupler\cite{haeupler2014interactive:FOCS}, except for the fact that the number of hash collisions can be bounded by $\Theta(n\delta)$. 
This is because of the fact that the entire Haeupler \cite{haeupler2014interactive:FOCS} analysis, except Lemma 7.7 that bounds hash collisions, are merely based on the fact that all but $O(n\delta)$ rounds are error free. 

Therefore, if we show that the number of hash collisions in the fully adversarial case is bounded by $\Theta(n\delta)$, the simulation rate will be 
$$1 - \left(\frac{\Theta(n\delta)}{n/r} + \frac{r_c}{\bar r}\right) = 1 - \Theta\left(\sqrt \frac{\log ({1}/{\delta})}{\delta} \cdot \delta + \frac{\log \log \frac{1}{\delta}}{\sqrt \frac{\log ({1}/{\delta})}{\delta}}\right) = 1 - \Theta\left(\sqrt{\delta \log \frac{1}{\delta} }\right)$$
and the proof will be complete.

We now bound the number of hash collisions in our interactive coding scheme. The verification process for fully adversarial setting uses a two level hash function $hash_2(hash_1(\cdot))$, where the seed of $hash_1$ is randomly generated in each round.

Lemma 7.7 of Haeupler \cite{haeupler2014interactive:FOCS} implies that for any oblivious adversary, the number of hash collisions due to $hash_1$ is at most $\Theta(n\delta)$ with probability $1-\delta^{\Theta(n\delta)}$. To find a similar bound for non-oblivious adversaries, we count the number of all possible oblivious adversaries and then use a union bound.
The number of oblivious adversaries is shown to be less than
$\left(\frac{3e}{\delta}\right)^{n\delta} = e^{n\delta\log(1/\delta)(1+o(1))}$
in Lemma~\ref{theorem:InterPreCodeFully}.
Hence, the probability of having more than $\Theta(n\delta)$ $hash_1$ hash collisions in any fully adversarial scenario is at most $2^{-\Theta(n\delta)}$. Now, Corollary 7.8 of Haeupler \cite{haeupler2014interactive:FOCS} can be directly applied to show that the number of $hash_2$ hash collisions can also be bounded by $\Theta(n\delta)$ with probability $1-2^{-\Theta(n\delta)}$. This completes the proof.
\end{proof}
}
\fullOnly{\ProofOfInterFullyAdv}

This insertion-deletion interactive coding scheme is, to the best of our knowledge, the first to be computationally efficient, to have communication rate approaching one, and to work over arbitrarily small alphabets.


\fullOnly{\section{Synchronization Strings and Edit-Distance Tree Codes}\label{sec:SynchAndTreeCodes}
We start by providing a new upper bound on the error tolerance of Braverman et al.'s coding scheme for interactive communication with a large alphabet over a insertion-deletion channel \cite{braverman2015coding}. We tweak the definition of edit-distance tree codes, the primary tool that Braverman et al. use in their coding scheme. In doing so, we show that their scheme has an error tolerance of $1/10 - \eps$ rather than $1/18-\eps$, which is the upper bound provided by Braverman et al. In particular, we prove the following theorem, which is a restatement of Theorem 1.4 from Braverman et al.'s work except for the revised error tolerance. The proof can be found in Section~\ref{sec:revisedDef}.

\begin{theorem}\label{thm:improvedBGMO}
For any $\eps > 0$, and for any binary (noiseless) protocol $\Pi$ with communication $CC(\Pi)$, there exists a noise-resilient coding scheme with communication $O_{\eps}(CC(\Pi))$ that succeeds in simulating $\Pi$ as long as the adversarial edit-corruption rate is at most $1/10-\eps$.
\end{theorem}

We then review the definition and key characteristics of edit-distance tree codes and  discuss the close relationship between edit-distance tree codes and synchronization strings using the revised definition of edit-distance tree codes.

\subsection{Revised definition of edit-distance tree codes}\label{sec:revisedDef}
Before proceeding to the definition of edit-distance tree codes, we begin with some definitions.

\begin{definition}[Prefix Code, Definition 3.1 of \cite{braverman2015coding}]
A prefix code $C:\Sigma_{in}^{n}\rightarrow\Sigma_{out}^{n}$ is a code such that $C(x)[i]$ only depends on $x[1,i]$. $C$ can also be considered as a $|\Sigma_{in}|$-ary tree of depth $n$ with symbols written on edges of the tree using the alphabet $\Sigma_{out}$. On this tree, each tree path from the root of length $l$ corresponds to a string $x \in \Sigma_{out}^l$ and the symbol written on the deepest edge of this path corresponds to $C(x)[l]$.
\end{definition}

\begin{definition}[$\eps$-bad Lambda, Revision of Definition 3.2 of \cite{braverman2015coding}]\label{def:bad_lambda}
We say that a prefix code $C$ contains an $\eps$-bad lambda if when this prefix code is represented as a tree, there exist four tree nodes $A,B,D,E$ such that $B \not= D$, $B \not= E$, $B$ is $D$ and $E$'s common ancestor, $A$ is $B$'s ancestor or $B$ itself, and $ED(AD, BE) \le (1-\eps) \cdot (|AD| + |BE|)$.\footnote{Braverman et al. say that $A,B,D,E$ form an $\eps$-bad lambda if $ED(AD, BE) \le (1-\eps) \cdot \max(|AD|, |BE|)$.} Here $AD$ and $BE$ are strings of symbols along the tree path from $A$ to $D$ and the tree path from $B$ to $E$. See Figure~\ref{fig:pathtree} for an illustration.
\end{definition}

\begin{definition}[Edit-distance Tree Code, Definition 3.3 of \cite{braverman2015coding}]
We say that a prefix code $C: \Sigma^n_{in} \rightarrow  \Sigma^n_{out}$ is an $\eps$-edit-distance tree code if $C$ does not contain an $\eps$-bad lambda.
\end{definition}



Using this revised definition of edit-distance tree code, we prove Theorem~\ref{thm:improvedBGMO} as follows:

\begin{proof}[Proof of Theorem~\ref{thm:improvedBGMO}]
We relax Braverman et al.'s definition of an \emph{edit-distance tree code} \cite{braverman2015coding}, which is an adaptation of Schulman's original tree codes \cite{schulman1993deterministic}. Edit-distance tree codes are parameterized by a real value $\eps$, and Braverman et al. define an $\eps$-edit-distance tree code to be a prefix code $C: \Sigma^n_{in} \rightarrow  \Sigma^n_{out}$ that does not contain what they refer to as an \emph{$\eps$-bad lambda.} The code $C$ contains an $\eps$-bad lambda if when one considers $C$ as a tree, there exist four tree nodes $A,B,D,E$ such that $B \not= D$, $B \not= E$, $B$ is $D$ and $E$’s common ancestor, $A$ is $B$’s ancestor or $B$ itself, and $ED(AD, BE) \le (1-\eps) \cdot \max(|AD|, |BE|)$. Here $AD$ and $BE$ are strings of symbols along the tree path from $A$ to $D$ and the tree path from $B$ to $E$. See Figure~\ref{fig:pathtree} for an illustration. We relax the definition of an $\eps$-bad lambda to be any four tree nodes $A,B,D,E$ as above such that $ED(AD, BE) \le (1-\eps) \cdot (|AD|+ |BE|)$.

We use synchronization strings together with Schulman's original tree codes to prove that edit-distance tree codes, under this revised definition, exist. In particular, in Theorem~\ref{thm:syncandtreecodes}, we show that synchronization strings can be concatenated with tree codes to produce edit-distance tree codes. Given that tree codes and synchronization strings exist, this means that edit-distance tree codes according to our revised definition exist.

As we saw in Lemma~\ref{lemma:treecodes_SD}, if $C: \Sigma^n_{in} \rightarrow \Sigma^n_{out}$ is an $\eps$-edit-distance tree code and $\tilde{c} \in \Sigma^m_{out}$, then there exists at most one $c \in \cup_{i = 1}^n C\left(\Sigma_{in}^n\right)[1,i]$ such that $SD(c, \tilde{c}) \leq 1-\eps$. This leads to an improved version of Lemma 4.2 from Braverman et al.'s work, which we describe after we set up notation. In Braverman et al.'s protocol, let $N_A$ and $N_B$ be the number of messages Alice and Bob have sent when one of them reaches the end of the protocol. Let $s_A$ and $r_A$ be the messages Alice has sent and received over the course of the protocol, and let $s_B$ and $r_B$ be the messages Bob has sent and received over the course of the protocol.  Let $\tau_A = (\tau_1, \tau_2)$ be the string matching between $s_B[1,N_B]$ and $r_A[1, N_A]$ and let $\tau_B = (\tau_3, \tau_4)$ be the string matching between $s_A[1, N_A]$ and $r_B[1, N_B]$. Braverman et al. call a given round a \emph{good decoding} if Alice correctly decodes the entire set of tree code edges sent by Bob. Lemma~\ref{lemma:treecodes_SD} admits the following improved version of Lemma 4.2 from Braverman et al's work:

\begin{lemma}[Revision of Lemma 4.2 from \cite{braverman2015coding}]\label{lem:improved4.2}
Alice has at least $N_A + \left(1 - \frac{1}{1-\eps}\right)sc(\tau_2) - \left(1 + \frac{1}{1-\eps}\right)sc(\tau_1)$ good decodings. Bob has at least $N_B + \left(1 - \frac{1}{1-\eps}\right)sc(\tau_4) - \left(1 + \frac{1}{1-\eps}\right)sc(\tau_3)$ good decodings.
\end{lemma}

The proof of this revised lemma follows the exact same logic as the proof of Lemma 4.2 in Braverman et al.'s work. The theorem statement now follows the same logic as the proof of Theorem 4.4 in Braverman et al.'s work. We include the revised section of the proof below, where the only changes are the maximum error tolerance $\rho$, the value of $n$ (which is smaller by a factor of $\frac{1}{2}$), and several equations. We set $\rho = \frac{1}{10}-\eps$ and $n = \lceil\frac{T}{8\eps}\rceil$, where $T$ is the length of the original protocol. Let $g_A$ be the number of good decodings of Alice and $b_A = N_A - g_A$ be the number of bad decodings of Alice. Similarly, let $g_B$ be the number of good decodings of Bob and $b_B = N_B - g_B$ be the number of bad decodings of Bob. In Braverman et al.'s protocol, Alice and Bob share an edit-distance tree code $C: \Sigma^n_{in} \rightarrow  \Sigma^n_{out}$. By definition of edit distance, $sc(\tau_1) + sc(\tau_3) = sc(\tau_2) + sc(\tau_4) \leq 2\rho n$. By Lemma~\ref{lem:improved4.2}, \begin{align*}
b_A + b_B &\leq \frac{1}{1-\eps}(sc(\tau_1) + sc(\tau_2)) + sc(\tau_1) - sc(\tau_2) + \frac{1}{1-\eps}(sc(\tau_3) + sc(\tau_4)) + sc(\tau_3) - sc(\tau_4)\\
&\leq \frac{4 \rho n}{1-\eps}.
\end{align*} Therefore,
\begin{align*}
g_A &= N_A - b_A \geq N_A - \frac{4\rho n}{1-\eps} = (N_A - n(1-2\rho) + n(1-2\rho) + n(1-2\rho) - \frac{4\rho n}{1-\eps}\\
& \geq b_A + b_B - \frac{4 \rho n}{1-\eps} + (N_A - n(1-2\rho))  + n(1-2\rho) - \frac{4\rho n}{1-\eps}\\
& = b_A + b_B + (N_A - n(1-2\rho)) + n\left(1- \frac{8\rho}{1-\eps} - 2\rho\right)\\
& \geq b_A + b_B + (N_A - n(1-2\rho)) + n\left(1- \frac{8\rho}{1-\eps} - 2\rho\right)\\
& \geq b_A + b_B + (N_A - n(1-2\rho)) + 8\eps n \geq b_A + b_B + (N_A - n(1-2\rho)) + T.
\end{align*} The remainder of the proof follows exactly as is the proof of Theorem 4.4 from Braverman et al.'s work.
\end{proof}

Suppose that Alice uses an $\eps$-edit-distance tree code $C: \Sigma^n_{in} \rightarrow \Sigma^n_{out}$ to send a message $c$ to Bob. In the following theorem, we show that if $SD(c,\tilde{c}) < 1 - \eps$, then Bob can decode Alice's message. Braverman et al. proved a weaker version of this theorem \cite{braverman2015coding}, but our revised definition of an edit-distance tree code easily admits the follow lemma.

\begin{lemma}[Revision of Lemma 3.10 from \cite{braverman2015coding}]\label{lemma:treecodes_SD}
Let $C: \Sigma^n_{in} \rightarrow \Sigma^n_{out}$ be an $\eps$-edit-distance tree code, and let $\tilde{c} \in \Sigma^m_{out}$. There exists at most one $c \in \cup_{i = 1}^n C\left(\Sigma_{in}^n\right)[1,i]$ such that $SD(c, \tilde{c}) \leq 1-\eps$.
\end{lemma}

\begin{proof}
The proof follows exactly the same logic as the proof of Lemma 3.10 from \cite{braverman2015coding}. The proof begins by assuming, for a contradiction, that there exist two messages $c,c' \in \cup_{i = 1}^n C\left(\Sigma_{in}^n\right)[1,i]$ such that both $SD(c, \tilde{c}) \leq (1-\eps)$ and $SD(c', \tilde{c}) \leq (1-\eps)$. Then, the proof reaches its conclusion by the exact same logic.
\end{proof}



In Theorem~\ref{thm:extension}, we will show that one can extend an $\eps$-synchronization string $S$ to a $(2\eps-\eps^2)$-edit-distance tree code. This implies the following corollary of Lemma~\ref{lemma:treecodes_SD}.

\begin{corollary}
Let $S\in\Sigma^n$ be an $\eps$-synchronization string, and $c \in \Sigma^m$. There exists at most one prefix $\tilde{c}$ of $S$ for which $SD(c, \tilde{c})\le (1-\eps)^2$.
\end{corollary}

\subsection{Edit-distance tree codes and synchronization strings}

We prove that under the revised definition, edit-distance tree codes still exist by relating edit-distance tree codes to Schulman's original tree codes. We introduce a method to obtain an edit-distance tree code using ordinary tree codes and synchronization strings. Intuitively, using synchronization strings and tree codes together, one can overcome the synchronization problem using the synchronization string and then overcome the rest of the decoding challenges using tree codes, which are standard tools for overcoming corruption errors. Tree codes are defined as follows.

\begin{definition}[Tree Codes, from \cite{schulman1996Interactive}]
A $d$-ary tree code over an alphabet $\Sigma$ of distance parameter $\alpha$ and depth $n$ is a $d$-ary tree of depth $n$ in which every arc of the tree is labeled with a character from the alphabet $\Sigma$ subject to the following condition. Let $v_1$ and $v_2$ be any two nodes at some common depth $h$ in the tree. Let $h - l$ be the depth of their least common ancestor. Let $W(v_1)$ and $W(v_2)$ be the concatenation of the letters on the arcs leading from the root to $v_1$ and $v_2$ respectively. Then $\Delta(W(v_1), W(v_2)) \ge \alpha l$.
\end{definition}

Schulman \cite{schulman1993deterministic,schulman1996Interactive} proved that for any $\alpha < 1$, there exists a $d$-ary tree code with distance $\alpha$ and infinite depth over an alphabet of size $(cd)^{1/(1-\alpha)}$, for some constant $c < 6$. We now prove that edit-distance tree codes can be obtained using synchronization strings and tree codes.

\begin{theorem}\label{thm:syncandtreecodes}
Let $T$ be a tree code of depth $n$ and distance parameter $\left(1-\alpha\right)$ and let $S$ be an $\eps$-synchronization string of length $n$. Let the tree $T'$ be obtained by concatenating all edges on $i^{th}$ level of $T$ with $S[i]$. Then $T'$ is a $(1-\eps-\alpha)$-edit-distance tree code.
\end{theorem}
\begin{proof}
To prove that $T'$ is a $(1-\eps-\alpha)$-edit-distance tree code, we show that $T'$ does not contain any $(1-\eps-\alpha)$-bad lambdas. Let $A, B, D$, and $E$ be nodes in $T'$ that form a lambda structure. 

We refer to the string consisting of edge labels on the path from node $A$ to $D$ by $AD$ and similarly we refer to the string consisting of labels on the path from $B$ to $E$ by $BE$. Let $\tau$ be the string matching characterizing the edit distance of $AD$ and $BE$. We call $i$ a \emph{matching position} if $\tau_1[i] = \tau_2[i]$. Further, we call $i$ a \emph{same-level matching position} if $\tau_1[i]$ and $\tau_2[i]$ correspond to symbols on the same level edges and \emph{disparate-level matching position} otherwise.


By the definition of a tree code, the number of same-level matching positions is at most 
\[\alpha \min\left\{|BD|, |BE|\right\} \leq \alpha (|AD| + |BE|).\]
Further, Haeupler and Shahrasbi prove that for any monotone matching between an $\eps$-synchronization string and itself, i.e. a set of pairs $M = \left\{\left(a_1, b_1\right), \dots, \left(a_m, b_m\right)\right\}$ where $a_1 < \cdots < a_m$, $b_1 < \cdots b_m$, and $S\left[a_i\right] = S\left[b_i\right]$, the number of pairs where $a_i \not= b_i$ is at most $\eps |S|$ \cite{haeupler2017synchronization}.
This means that the number of disparate-level matching positions is at most $\eps(|AD|+|BE|)$.
Therefore, 
\begin{eqnarray*}
ED(AD, BE) &\ge& |AD|+|BE| - \eps(|AD|+|BE|) - \alpha (|AD|+|BE|)\\ 
&\ge& (1-\eps-\alpha)(|AD|+|BE|).
\end{eqnarray*}
\end{proof}

Note that Theorem~\ref{thm:syncandtreecodes} suggests a construction for edit-distance tree codes by simply constructing and concatenating an ordinary tree code and a synchronization string. As synchronization strings are efficiently constructible and tree codes can be constructed in sub-exponential time~\cite{Schulman03:Postscript, braverman2012towards}, this construction runs in sub-exponential time which improves over the construction of edit-distance tree codes from Braverman et al.~\cite{braverman2015coding}.

The following theorems discuss further connections between synchronization strings and edit-distance tree codes. 
Theorems~\ref{thm:TreeToSynch} and \ref{thm:extension} show that edit-distance tree codes and synchronization strings are essentially similar combinatorial objects.
In Theorem~\ref{thm:TreeToSynch}, we show that the edge labels along a monotonically down-going path in an edit-distance tree code form a synchronization string as though, in a manner, synchronization strings are one dimensional edit-distance tree codes. On the other hand, in Theorem~\ref{thm:extension}, we show that for any synchronization string $S$, there is an edit-distance tree code that has $S$ on one of its monotonically down-going paths.

\begin{lemma}\label{lem:TreeToSynch}
In an $\eps$-edit-distance tree code, for every three vertices $X$, $Y$, and $Z$ where $X$ is an ancestor of $Y$ and $Y$ is an 
ancestor of $Z$, we have that $ED(XY,YZ) \ge \left(1 - \eps - \frac{1}{|XZ|}\right) |XZ|$.
\end{lemma}
\begin{proof}
To show this, we set $A=X, B = Y, E=Z$ and denote the child of $Y$ which is not in the path from $Y$ to $Z$ by $D$ (see Figure \ref{fig:pathtree}). Then $A, B, D$, and $E$ form a lambda in the tree (Definition~\ref{def:bad_lambda}). As $\eps$-edit-distance tree codes are $\eps$-bad lambda free, we know that $ED(AD, BE) \ge (1-\eps) \cdot(|AD|+|BE|)$.
\begin{figure}
\centering
\includegraphics[scale=1.2]{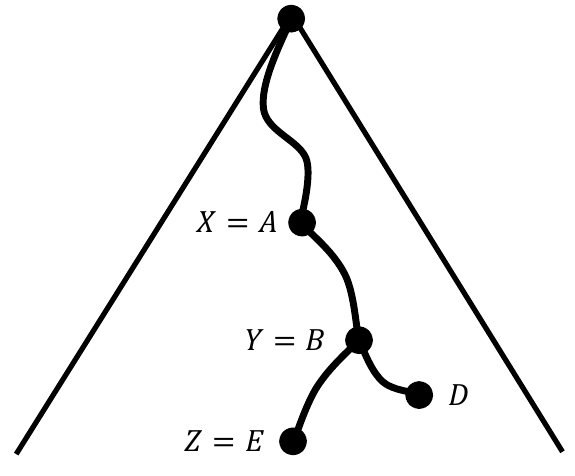}
\caption{$(A, B, D, E)$ form a Lambda structure.}
\label{fig:pathtree}
\end{figure}

Note that $ED(AD, BE) = ED(XD, YZ) \le 1 + ED(XY, YZ)$ and $|AD| = |XY| + 1 > |XY|$, which means that
$ED(XY,YZ) + 1 \ge (1-\eps) |XZ|.$
Therefore, $ED(XY,YZ) \ge \left(1 - \eps-\frac{1}{|XZ|}\right)\cdot |XZ|$.
\end{proof}
Using this property, one can obtain synchronization strings using monotonically down-going paths in a given edit-distance tree code as follows:
\begin{theorem}\label{thm:TreeToSynch}
Concatenating the symbols on any monotonically down-going path in an $\eps$-edit-distance tree code with a string consisting of repetitions of $1,2,\cdots, l$ gives an $\left(\eps + \frac{1}{l}\right)$-synchronization string.
\end{theorem}
\begin{proof}
Consider three indices $i < j < k$ in such string. If $k - i < l$, the $\left(\eps + \frac{1}{l}\right)$-synchronization property holds as a result of the $1,\cdots, l,\cdots$ portion. Unless, the edit-distance path symbols satisfy $\left(\eps + \frac{1}{l}\right)$-synchronization property as a result of Lemma~\ref{lem:TreeToSynch}.
\end{proof}

Theorem~\ref{thm:TreeToSynch} results into the following corollary:
\begin{corollary}
Existence of synchronization strings over constant-sized alphabets can be implied by Theorem~\ref{thm:TreeToSynch} and the fact that edit-distance tree codes exist from~\cite{braverman2015coding}. However, the alphabet size would be exponentially large in terms of $\frac{1}{\eps}$.
\end{corollary}

\begin{theorem}\label{thm:extension}
Any $\eps$-synchronization string in $\Sigma_1^n$ can be extended to a $(2\eps - \eps^2)$-edit-distance tree code $C:\Sigma_{in}\rightarrow\left(\Sigma_{out}\cup\Sigma_1\right)^n$ using the $\eps$-edit-distance tree code $C':\left(\Sigma_{in}^n\rightarrow\Sigma_{out}^n\right)$ such that a monotonically down-going path on $C$ is labeled as $S$. 
\end{theorem}
\begin{proof}
We simply replace the labels of edges on the rightmost path in the tree associated with $C'$ to obtain $C$. We show that $C$ is a valid $(2\eps - \eps^2)$-edit-distance tree code. To prove this, we need to verify that this tree does not contain an $(2\eps - \eps^2)$-bad lambda structure. In this proof, we set $\alpha = 1 - (2\eps - \eps^2) = (1-\eps)^2$ and $\alpha' = 1 - \eps$.

Let $A, B, D, E$ form a lambda structure in the corresponding tree of $C$. 
We will prove the claim under two different cases:
\begin{itemize}
\item \textbf{$BE$ does not contain any edges from the rightmost path:} 
Assume that $AD$ contains $l$ edges from the rightmost path. In order to turn $AD$ into $BE$, we must remove all $l$ symbols associated with the rightmost path since BE has no symbols from $\Sigma_1$. Hence, $ED(AD, BE)$ is equal to $l$ plus the edit distance of $BE$ and $AD$ after removal of those $l$ edges. Note that removing $l$ elements in $AD$ and then converting the remaining string ($\widetilde{AD}$) to $BE$ is also a way of converting $AD$ to $BE$ in the original tree. Thus,
$$ED_{original}(AD, BE) \le l + ED(\widetilde{AD}, BE) = ED_{modified}(AD+ BE)$$
On the other hand, $ED_{original}(AD, BE) \ge \alpha'\cdot(|AD|+ |BE|).$
Therefore, \[ED_{modified}(AD, BE) \ge \alpha'\cdot(|AD|+ |BE|).\]
\item \textbf{$BE$ contains some edges from the rightmost path:} 
It must be that all of $AB$ and a non-empty prefix of $BE$, which we refer to as $BX$, both lie in the rightmost path in tree code (Fig. \ref{fig:distance}).
Since the symbols in the rightmost path are from the alphabet $\Sigma_1$,
$$ED(AD, BE) = ED(AB, BX) + ED(BD, XE).$$
We consider the following two cases, where $c$ is a constant we set later.
\begin{enumerate}
\item $c\cdot(|AB| + |BX|) > |BD| + |XE|$:

In this case,
\begin{eqnarray*}
ED(AD, BE) &=& ED(AB, BX) + ED(BD, XE)\\
&\ge& ED(AB, BX)\\
&\ge& (1-\eps)\cdot (|AB| + |BX|)\\
&\ge& \frac{1-\eps}{c+1}\cdot \left((|AB|+ |BX|)+(|BD|+|EX|)\right)\\
&\ge& \frac{1-\eps}{c+1}\cdot (|AD|+ |BE|).
\end{eqnarray*} 
\item $c\cdot(|AB|+ |BX|) \le |BD|+|XE|$:
Since $(A, B, D, E)$ form a lambda structure, \[ED_{original}(AD, BE) \ge \alpha'\cdot(|AD|+ |BE|)\]
and therefore,
\begin{eqnarray*}
ED_{modified}(AD, BE) &\ge& ED_{modified}(BD, EX) + ED(AB, BX)\\
&\ge& \left[ED_{original}(AD, BE) - (|AB|+|BX|)\right] + ED(AB, BX)\\
&\ge& \alpha'\cdot(|AD|+ |BE|) - (|AB|+|BX|) + (1-\eps)\cdot(|AB|+|BX|)\\
&=& \alpha'\cdot(|AD|+ |BE|) - \eps(|AB|+|BX|)\\
&\ge& \alpha'\cdot(|AD|+ |BE|) - \frac{\eps}{c+1}(|AD|+|BE|)\\
&=&\left(\alpha' - \frac{\eps}{c+1}\right)(|AD|+|BE|).
\end{eqnarray*}
\end{enumerate}
\end{itemize}
Hence:
$$ED(AD, BE) \ge \min\left\{\alpha' - \frac{\eps}{c+1}, \frac{1-\eps}{c+1},\alpha' \right\}(|AD|+|BE|)$$
As $\alpha'=1-\eps$, setting
$c = \frac{\eps}{1-\eps}$ gives that 
$$ED(AD, BE) \ge (1-\eps)^2\cdot(|AD|+|BE|)$$ which finishes the proof.
\begin{figure}
\centering
\includegraphics[scale=.90]{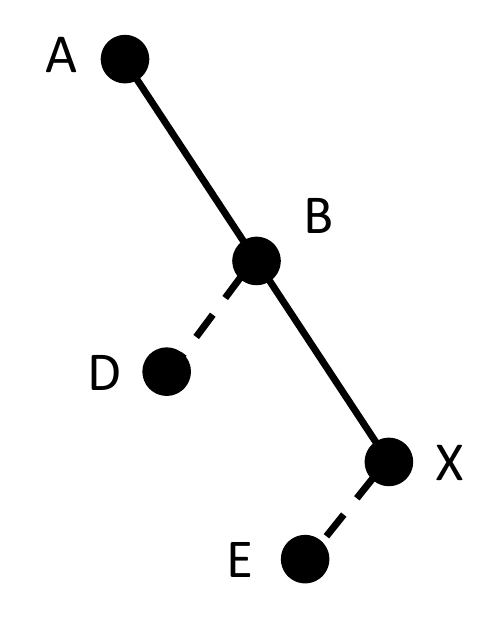}
\caption{$AD$ and $BE$ both contain edges from the rightmost path. Straight lines represent the edges in the rightmost path and dashed ones represent other edges.}
\label{fig:distance}
\end{figure}
\end{proof}
}
\section*{Acknowledgements}
The authors thank Allison Bishop for valuable discussions in the early stages of this work. 

\newpage
\bibliographystyle{plain}
\bibliography{bibliography}

\shortOnly{\newpage
\input{5-Appendix}}

\end{document}